\documentclass[11pt]{article}
\usepackage{epsfig}
\usepackage{amssymb,amsmath,amsthm,url}
\usepackage{multirow}
\usepackage{booktabs}
\usepackage{setspace}
\usepackage{dsfont}

\usepackage{times,amsfonts,eucal}
\usepackage{algorithm}
\usepackage{algorithmic}
\usepackage{graphicx,ifthen}
\usepackage{wrapfig}
\usepackage{latexsym}
\usepackage{color}
\usepackage{mathrsfs}
\usepackage{bm}
\usepackage{bbm}
\usepackage{srctex}

\usepackage{bm}
\newtheorem{theorem}{Theorem}[section]

\newtheorem{cor}{Corollary}[section]
\newtheorem{definition}{Definition}

\newtheorem{assumption}{Assumption}[section]
\newtheorem{algo}{Algorithm}[section]
\numberwithin{equation}{section}
\numberwithin{theorem}{section}
\numberwithin{lemma}{section}
\numberwithin{pro}{section}
\numberwithin{cor}{section}
\numberwithin{definition}{section}
\numberwithin{cons}{section}
\numberwithin{rem}{section}
\numberwithin{exa}{section}
\numberwithin{table}{section}
\numberwithin{figure}{section}
\numberwithin{algo}{section}

\newcommand{\Z}{\mathbf{Z}}

\def\N{{\mathbb N}}

\def\Z{{\mathbb Z}}

\def\E{{\rm E}}
\def\R{{\mathbb R}}

\def\cals_+{{\cals_+}}
\def\calb{{\mathcal{B}}}

\def\cala{{\mathcal{A}}}

\def\call{{\mathcal{L}}}
\def\cals{{\mathcal{S}}}

\def\caln{{\mathcal{N}}}

\newcommand{\wh}{\hat}

\newcommand{\spa}{{\rm \overline{sp}}}

\def\beq{\begin{equation}}
\def\eeq{\end{equation}}
\def\bals{\begin{align*}}
\def\eals{\end{align*}}

\def\bal{\begin{align}}
\def\eal{\end{align}}

\allowdisplaybreaks

\begin{document}

\title{%On the estimation of functional moving average processes
Estimating functional time series by moving average model fitting\footnote{This research was partially supported by NSF grants DMS 1305858 and DMS 1407530}
}

\author{
Alexander Aue\footnote{Department of Statistics, University of California, Davis, CA 95616, USA, email: \tt{aaue@ ucdavis.edu}}
\and Johannes Klepsch\footnote{Center for Mathematical Sciences, Technische Universit\"at M\"unchen,  85748 Garching, Boltzmannstra{\ss}e 3, Germany, email: \tt{j.klepsch@tum.de}}\;\footnote{Corresponding author}
}

\date{\today}
\maketitle
\bibliographystyle{plain}

\begin{abstract}
\setlength{\baselineskip}{1.8em}
Functional time series have become an integral part of both functional data and time series analysis. Important contributions to methodology, theory and application for the prediction of future trajectories and the estimation of functional time series parameters have been made in the recent past. This paper continues this line of research by proposing a first principled approach to estimate invertible functional time series by fitting functional  moving average processes. The idea is to estimate the coefficient operators in a functional linear filter. To do this a functional Innovations Algorithm is utilized as a starting point to estimate the corresponding moving average operators via suitable projections into principal directions. In order to establish consistency of the proposed estimators, asymptotic theory is developed for increasing subspaces of these principal directions. For practical purposes, several strategies to select the number of principal directions to include in the estimation procedure as well as the choice of order of the functional moving average process are discussed. Their empirical performance is evaluated through simulations and an application to vehicle traffic data. \medskip 

\end{abstract}

\setlength{\baselineskip}{1.8em}

%%%%%%%%%%%%%%%%%%%%%%%%
\section{Introduction}
\label{sec:intro}
%%%%%%%%%%%%%%%%%%%%%%%%

With the advent of complex data came the need for %statistical
methods to address %the 
novel statistical challenges. Among the new methodologies, functional data analysis  provides a particular set of tools for tackling questions related to observations conveniently viewed as entire curves rather than individual data points. The current state of the field may be reviewed in one of the comprehensive monographs written by Bosq \cite{bosq}, Ramsay and Silverman \cite{ramsay1}, Horv\'ath and Kokoszka \cite{horvath}, and Hsing and Eubank \cite{hsing}. Many of the applications discussed there point to an intrinsic time series nature of the underlying curves. This has led to an upsurge in contributions to the functional time series literature. The many recent works in this area include papers on time-domain methods such as H\"ormann and Kokoszka \cite{weaklydep}, who introduced a framework to describe weakly stationary functional time series, and Aue et al.\ \cite{aue} and Klepsch and Kl\"uppelberg \cite{kk}, who developed functional prediction methodology; as well as frequency domain methods such as Panaretos and Tavakoli \cite{panaretos}, who utilized functional cumulants to justify their functional Fourier analysis, H\"ormann et al.\ \cite{hoermann}, who defined the concept of dynamic functional principal components, and Aue and van Delft \cite{avd}, who designed stationarity tests based on functional periodogram properties. 

This paper is concerned with functional moving average (FMA) processes as a building block to estimate potentially more complicated functional time series. Together with the functional autoregressive (FAR) processes, the FMA processes comprise one of the basic functional time series model classes. They are used, for example, as a building block in the $L^p$-$m$-approximability concept of H\"ormann and Kokoszka \cite{weaklydep}, which is based on the idea that a sufficiently close approximation with truncated linear processes may adequately capture more complex dynamics, based on a causal infinite MA representation. It should be noted that, while there is a significant number of papers on the use of both FMA and FAR processes, the same is not the case for the more flexible functional autoregressive moving average (FARMA) processes. This is due to the technical difficulties that arise from transitioning from the multivariate to the functional level. One advantage that FMA processes enjoy over other members of the FARMA class is that their projections remain multivariate MA processes (of potentially lower order). This is one of the reasons that makes them attractive for further study.

Here interest is in estimating the dynamics of an invertible functional linear process through fitting FMA models. The operators in the FMA representation, a functional linear filter, are estimated using a functional Innovations Algorithm. This counterpart of the well-known univariate and multivariate Innovations Algorithms was recently introduced by Klepsch and Kl\"uppelberg \cite{kk}, where its properties were analyzed on a population level. These results are extended to the sample case and used as a first step in the estimation. The proposed procedure uses projections to a number of principal directions, estimated through functional principal components analysis (see, for example, Ramsay and Silverman \cite{ramsay1}). To ensure appropriate large-sample properties of the proposed estimators, the dimensionality of the principle directions space is allowed to grow slowly with the sample size. In this framework, the consistency of the estimators of the functional linear filter is the main theoretical contribution. It is presented in Section \ref{sec:methodology}.

The theoretical results are accompanied by selection procedures to guide the selection of the order of the approximating FMA process and the dimension of the subspace of principal directions. To choose the dimension of the subspace a sequential test procedure is proposed. Order selection based on AICC, Box--Ljung  and FPE type criteria are suggested. Details of the proposed model selection procedures are given in Section \ref{sec:selection}. Their practical performance is highlighted in Section~\ref{sec:sim}, where results of a simulation study are reported, and Section~\ref{sec:app}, where an application to real-world data on vehicle traffic data is discussed. 

To summarize, this paper is organized as follows. Section \ref{sec:setting} briefly reviews basic notions of Hilbert-space valued random variables before introducing the setting and the main assumptions. The proposed estimation methodology for functional time series is detailed in Section \ref{sec:methodology}. Section \ref{sec:selection} discusses in some depth the practical selection of the dimension of the projection space and the order of the approximating FMA process. These suggestions are tested in a Monte Carlo simulation study and an application to traffic data in Sections \ref{sec:sim} and \ref{sec:app}, respectively. Section \ref{sec:conclusion} concludes and proofs of the main results can be found in Section \ref{sec:proof}.

%%%%%%%%%%%%%%%%%%%%%%%%
\section{Setting}
\label{sec:setting}
%%%%%%%%%%%%%%%%%%%%%%%%

Functional data is often conducted in $H=L^2[0,1]$, the Hilbert-space of square-integrable functions, with canonical norm $\|x\|=\langle x,x\rangle^{1/2}$ induced by the inner product $\left\langle x , y \right\rangle=\int_0^1 x(s)y(s)ds$ for $x,y\in H$. For an introduction to Hilbert spaces from a functional analytic perspective, the reader is referred to Chapters~3.2 and 3.6 in Simon~\cite{simon}. All random functions considered in this paper are defined on a probability space $(\Omega,\mathcal{A},\mathcal{P})$ and are assumed to be $\cala$-$\calb_H$-measurable, where $\calb_H$ denotes the Borel $\sigma$-algebra of subsets of $H$. 
Note that the space of square integrable random functions  $L^2_{H}=L^2(\Omega, \mathcal{A},\mathcal{P})$ is a Hilbert space with inner product $\E[\left\langle X,Y\right\rangle]=\E[\int_0^1X(s)Y(s)ds]$ for $X,Y \in L^2_{H}$. Similary, denote by $L^p_H=L^p(\Omega,\cala,\mathcal{P})$ the space of $H$-valued functions such that $\nu_p(X)=(\E[\|X\|^p])^{1/p}<\infty$. Let $\mathbb{Z}$, $\N$ and $\N_0$ denote the set of integers, positive integers and non-negative integers, respectively.

Interest in this paper is in fitting techniques for functional time series $(X_j\colon j\in\mathbb{Z})$ taking values in $L_H^2$. To describe a wide variety of temporal dynamics, the framework is established for %functional linear processes . The notion of $L^p$-$m$-approximability is adopted. A version of this notion was used for multivariate time series in Aue et al.\ \cite{AHHR} and then translated to the functional domain by H\"ormann and Kokoszka \cite{weaklydep}. The definition is as follows. Let $\mathbb{Z}$ denote the set of integers.
%
%\begin{definition}
%\label{def:lpm}
%{\rm 
%Let $p\geq 1$. A sequence $(X_j\colon j\in\mathbb{Z})$ with values in $L^p_H$ is called {\it $L^p$-$m$-approximable}\/ if
%\[
%X_j=f(\varepsilon_j,\varepsilon_{j-1},\ldots),
%\qquad j\in\mathbb{Z},
%\]
%can be represented as a functional Bernoulli shift with a sequence of independent, identically distributed random elements $(\varepsilon_j\colon j\in\mathbb{Z})$ taking values in the measurable space $S$ and a measurable function $f\colon S^\infty\to H$ such that 
%\[
%\sum_{m=0}^\infty\big(\E[\|X_j-X_{j,m}\|^p]\big)^{1/p},
%\] 
%where $X_{j,m}=f(\varepsilon_j,\ldots,\varepsilon_{j-m+1},\varepsilon_{j,m,j-m}^*,\varepsilon_{j,m,j-m-1}^*,\ldots)$ with $\varepsilon_{j,m,i}^*$ being independent copies of $\varepsilon_{j,0}$ independent of $(\varepsilon_j\colon j\in\mathbb{Z})$.
%}
%\end{definition}
%
%Conditions can be established for most of the common linear and nonlinear functional time series models to be $L^p$-$m$-approximable. In particular, 
functional linear processes $(X_j\colon j\in\mathbb{Z})$ defined through the series expansion
\begin{equation}
\label{eq:flp}
X_j=\sum_{\ell=0}^\infty \psi_\ell\varepsilon_{j-\ell},
\qquad j\in\mathbb{Z},
\end{equation}
%are naturally included if the condition $\sum_{m=1}^\infty\sum_{\ell=m}^\infty\|\psi_\ell\|_{\call}<\infty$ is met (see Proposition 2.1 in H\"ormann and Kokoszka \cite{weaklydep}). In \eqref{eq:flp}, 
where $(\psi_\ell\colon\ell\in\mathbb{N}_0)$ is a sequence in $\call$, the space of bounded linear operators acting on $H$, equipped with the standard norm $\|A\|_{\call}=\sup_{ \|x\|\leq 1} \|Ax\|$, and $(\varepsilon_j\colon j\in\mathbb{Z})$ is assumed to be an independent and identically distributed sequence in $L^2_H$. Additional summability conditions are imposed on the sequence of coefficient operators $(\psi_\ell\colon\ell\in\mathbb{N}_0)$ if it is necessary to control the rate of decay of the temporal dependence. %The latter is in contrast to standard time series literature (see, for example, Brockwell and Davis~\cite{brockwell}) that requires the innovations to be a white noise sequence. This is because the $L^p$-$m$-approximability concept utilizes the notion of strict rather than weak stationarity. 
Whenever the terminology ``functional linear process'' is used in this paper it is understood to be in the sense of \eqref{eq:flp}. Note that, as for univariate and multivariate time series models, every stationary causal functional autoregressive moving average (FARMA) process is a functional linear process (see Spangenberg \cite{spangenberg}, Theorem~2.3). Special cases include functional autoregressive processes of order $p$, FAR$(p)$, which have been thoroughly investigated in the literature, and the {\em functional moving average process of order $q$}, FMA$(q)$, which is given by the equation
\begin{align}
	X_j = \sum_{\ell=1}^q \theta_\ell \varepsilon_{j-\ell}+\varepsilon_j, \qquad j \in \mathbb{Z},  \label{FMA}
\end{align}
with $\theta_1,\ldots,\theta_q \in \call$. 

While the functional linear process in \eqref{eq:flp} is the prototypical causal time series, in the context of prediction, the concept of invertibility naturally enters; see Chapter 5.5 of Brockwell and Davis~\cite{brockwell}, and Nsiri and Roy~\cite{nsiri}. For a functional time series $(X_j\colon j\in\mathbb{Z})$ to be {\em invertible}, it is required that
\begin{equation}\label{eq:invertible}
X_j=\sum_{\ell=1}^\infty\pi_\ell X_{j-\ell}+\varepsilon_j,
\qquad j\in\mathbb{Z},
\end{equation}
for $(\pi_\ell\colon\ell\in\mathbb{N})$ in $\call$ such that $\sum_{\ell=1}^\infty\|\pi_\ell\|_{\call}<\infty$; see Merlev\`ede \cite{merlevede}. A sufficient condition for invertibility of a functional linear process, which is assumed throughout, is given in Theorem~7.2 of Bosq \cite{bosq}.

The definition of a functional linear process in \eqref{eq:flp} provides a convenient framework for the formulation of large-sample results and their verification. In order to analyze time series characteristics in practice, however, most statistical methods require a more in-depth understanding of the underlying dependence structure. This is typically achieved through the use of autocovariances which determine the second-order structure. Observe first that any random variable in $L^p_H$ with $p\geq 1$ possesses a unique {\em mean function} in $H$, which allows for a pointwise definition; see Bosq \cite{bosq}. For what follows, it is assumed without loss of generality that $\mu=0$, the zero function. If $X\in L_H^p$ with $p\geq 2$ such that $\E[X]=0$, then the {\em covariance operator} of $X$ exists and is given by 
\begin{align*}
	C_X(y) = \E [\langle X, y \rangle X], \qquad y \in H.
\end{align*}
If $X,Y \in L^p_{H}$ with $p\geq 2$ such that $\E [X] = \E [Y]=0$, then the {\em cross covariance operator} of $X$ and $Y$ exists and is given by 
\begin{align*}
	C_{X,Y}(y) = C_{Y,X}^*(y)=\E [\langle X, y \rangle Y],  \qquad y \in H.
\end{align*}
where $C_{Y,X}^*$ denotes the adjoint of $C_{Y,X}$, noting that the adjoint $A^*$ of an operator $A$ is defined by the equality $\langle Ax,y \rangle = \langle x , A^* y\rangle$ for $x,y\in H$. The operators $C_X$ and $C_{Y,X}$ belong to $\caln$, the class of {\em nuclear operators}, whose elements $A$ have a representation $A=\sum_{j=1}^{\infty} \lambda_j\langle e_j , \cdot\rangle f_j $ with $\sum_{j=1}^{\infty} \vert \lambda_j \vert < \infty$ for two orthonormal bases (ONB) $(e_j)_{j\in\N}$ and $(f_j)_{j\in\N}$ of $H$. In that case $\Vert A\Vert_{\caln}=\sum_{j=1}^\infty \vert \lambda_j\vert <\infty$ %}$A$ are required to be bounded linear such that the nuclear norm $\| A\|_{\caln}=\sum_{i=1}^{\infty} \langle A e_i,e_i\rangle$ if finite for some (and hence all) orthonormal basis (ONB) $(e_i\colon i\in\N)$ of $H$
; see Section~1.5 of Bosq \cite{bosq}. Furthermore, $C_X$ is self-adjoint ($C_X=C_X^*$) and non-negative definite with spectral representation
\begin{align*}
	C_X(y)=\sum_{i=1}^{\infty} \lambda_i \langle y, \nu_i \rangle \nu_i, \qquad y\in H, 
\end{align*}
where $(\nu_i\colon i\in\N)$ is an ONB of $H$ and $(\lambda_i\colon i\in\N)$ is a sequence of positive real numbers such that $\sum_{i=1}^{\infty} \lambda_i < \infty$. When considering spectral representations, it is standard to assume that the $(\lambda_i\colon i\in\N)$ are ordered decreasingly and that there are no ties between consecutive $\lambda_i$. 

For ease of notation, introduce the operator 
%\begin{align*}
$x\otimes y(\cdot)= \langle x ,\cdot \rangle  y$ for $x,y\in H$. 
%\end{align*}
Then, $C_X = \E [  X\otimes X  ]$ and $C_{X,Y} = \E[X\otimes Y]$. Moreover, for a stationary process $(X_j\colon j\in\Z)$, the {\em lag-$h$ covariance operator} can be written as
\begin{align}
C_{X;h}=\E[X_0\otimes X_h],\qquad h\in \Z. \label{cxh}
\end{align}
The quantities in \eqref{cxh} are the basic building block in the functional Innovations Algorithm and the associated estimation strategy to be discussed in the next section.

%%%%%%%%%%%%%%%%%%%%%%%%
\section{Estimation methodology}
\label{sec:methodology}
%%%%%%%%%%%%%%%%%%%%%%%%

%This section is devoted to deriving consistent estimators for the parameters of invertible FMA processes. The proposed estimators are based on linear prediction equations, the idea being the same as in the multivariate case in \cite{lewis} and  \cite{mitchell}. We therefore start by deriving functional prediction equations. As general solutions to these equations do not exist on $H$, we propose to project our data on an appropriate finite dimensional space, where solutions to the prediction equations can be found. The parameter operators in the solution to the prediction equations provide initial estimates to the functional model parameters.  We investigate assumptions such that letting the dimension of the finite dimensional space tend to infinity yields consistent estimators. As it is done in \cite{bosq} for the estimation of FAR processes, we differentiate between the case where the finite dimensional space is known (Section~\ref{known}), and where it needs to be estimated (Section~\ref{unknown}).
%We start by recalling some linear prediction theory which will motivate our estimators.

\subsection{Linear prediction in function spaces}

Briefly recall the concept of linear prediction in Hilbert spaces as defined in Section~1.6 of Bosq \cite{bosq}. Let $(X_j\colon j\in \Z)$ be an invertible, functional linear process. Let $\bar{L}_{n,k}$ be the $\call$-closed subspace (LCS) generated by the stretch of functions $X_{n-k},\ldots,X_n$. LCS here is to be understood in the sense of Fortet \cite{fortet} that is $\bar{L}_{n,k}$ is the smallest subspace of $H$ containing $X_{n-k},\ldots,X_n$, closed with respect to operators in $\call$. Then, the best linear predictor of $X_{n+1}$ given $\{X_{n},X_{n-1},\dots,X_{n-k}\}$ at the population level is given by 
%by the projection of $X_{n+1}$ on some rich enough subspace containing $\{X_{n},X_{n-1},\dots,X_{n-k}\}$. 
%To characterize the subspace on which to project, the concept of $\call$-closed subspaces (LCS), introduced in \cite{fortet} is used. The LCS generated by $\{X_{n},X_{n-1},\dots,X_{n-k}\}$, that we denote by $\LCS(\{X_{n},X_{n-1},\dots,X_{n-k}\})$, is the smallest subspace containing all $X_i$, $i=n,\dots,n-k$, closed with respect to operators in $\call$. Hence, at the population level,  the best linear predictor of $X_{n+1}$ given $\{X_{n},X_{n-1},\dots,X_{n-k}\}$ is given by
\begin{align} \label{blp}
	\tilde X_{n+1,k}^f = P_{\bar{L}_{n,k}}(X_{n+1}),
	%P_{\LCS(\{X_{n},X_{n-1},\dots,X_{n-k}\})}(X_{n+1}),
\end{align}
where the superscript $f$ in the predictor notation indicates the fully functional nature of the predictor and $P_{\bar{L}_{n,k}}$ denotes projection on $\bar{L}_{n,k}$. Note that there are major differences to the multivariate prediction case. Due to the infinite dimensionality of function spaces, $\tilde X_{n+1,k}^f$ in \eqref{blp} is not guaranteed to have a representation in terms of its past values and operators in $\call$, see for instance Proposition~2.2 in Bosq \cite{bosq2014} and the discussion in Section~3 of Klepsch and Kl\"uppelberg \cite{kk}. A typical remedy in FDA is to resort to projections into principal directions and then to let the dimension $d$ of the projection subspace grow to infinity. At the subspace-level, multivariate methods may be applied to compute the predictors; for example the multivariate Innovations Algorithm; see Lewis and Reinsel \cite{lewis} and Mitchell and Brockwell \cite{mitchell}. This, however, has to be done with care, especially if sample versions of the predictors in \eqref{blp} are considered. Even at the population level, the rate at which $d$ tends to infinity has to be calibrated scrupulously to ensure that the inversions of matrices occurring, for example, in the multivariate Innovations Algorithm are meaningful and well defined (see Theorem 5.3 of Klepsch and Kl\"uppelberg \cite{kk}).

Therefore, the following alternative to the functional best linear predictor defined in \eqref{blp} is proposed. Recall that $(\nu_j\colon j\in\mathbb{N})$ are the eigenfunctions of the covariance operator $C_X$. Let $\mathcal{V}_d=\spa\{\nu_1,\dots,\nu_d\}$ be the subspace generated by the first $d$ principal directions and let $P_{\mathcal{V}_d}$ be the projection operator projecting from $H$ onto $\mathcal{V}_d$. Let furthermore $(d_i\colon i\in\N)$ be an increasing sequence of positive integers and define 
\begin{align}
 X_{d_i,j}=P_{\mathcal{V}_{d_i}}X_j,\qquad j\in\Z,\; i\in\N. \label{xdi}	
\end{align}
Note that \eqref{xdi} allows for the added flexibility of projecting different $X_j$ into different subspaces $\mathcal{V}_i$. Then, $X_{n+1}$ can be projected into the LCS generated by $X_{d_k,n},X_{d_{k-1},n-1},\ldots,X_{d_1,n-k}$, which is denoted by $\bar{{F}}_{n,k}$. Consequently, write
\begin{align}\label{blpd}
\tilde X_{n+1,k} = P_{\bar{{F}}_{n,k}}(X_{n+1})
%=\beta_{k,1}X_{d_k,n} + \beta_{k,2}X_{d_{k-1},n-1}+\dots+\beta_{k,k}X_{d_1,n-k+1}.
\end{align}
for the best linear predictor of $X_{n+1}$ given $\bar{{F}}_{n,k}$. This predictor could be computed by regressing $X_{n+1}$ onto $X_{d_k,n},X_{d_{k-1},n-1},\ldots,X_{d_1,n-k}$, but interest is here in the equivalent representation of $\tilde{X}_{n+1,k}$ in terms of one-step ahead prediction residuals given by
\begin{align}\label{innov}
	\tilde{X}_{n+1,k}
	= \sum_{i=1}^k\theta_{k,i}(X_{d_{k+1-i},n+1-i}-\tilde{X}_{n+1-i,k-i}),
	%= \theta_{k,1} (X_{d_k,n}- \tilde{X}_n^{k-1}) + \theta_{k,2}(X_{d_{k-1},n-1}- \tilde X^{k-2}_{n-1}) + \dots + \theta_{k,k}(X_{d_1,n-k}-\tilde{X}^0_{n-k}),
\end{align} 
where $\tilde{X}_{n-k,0}=0$. On a population level, it was shown in Klepsch and Kl\"uppelberg \cite{kk} that the coefficients $\theta_{k,i}$ with $k,i\in\N$ can be computed with the following algorithm.

\begin{algo}[\bf{Functional Innovations Algorithm}]
\label{fia}
Let $(X_j\colon j\in\Z)$ be a stationary functional linear process with covariance operator $C_X$ possessing eigenpairs $(\lambda_i,\nu_i\colon i\in\N)$ with $\lambda_i>0$ for all $i\in\N$. The best linear predictor $\tilde{ X}_{n+1,k}$ of ${ X}_{n+1}$  based on $\bar{{F}}_{n,k}$ defined in \eqref{innov} can be computed by the recursions
	\begin{align}
	\tilde{X}_{n-k,0}&=0\qquad\mbox{and}\qquad V_{1}=P_{\mathcal{V}_{d_1}}C_{X}P_{\mathcal{V}_{d_1}},\notag\\
	\tilde{X}_{n+1,k}&= \sum_{i=1}^{k} \theta_{k,i} (X_{d_{k+1-i},n+1-i}-\tilde{X}_{n+1-i,{k-i}}), \notag \\ %\label{xdhat1},\\
	\theta_{k,k-i}&=\bigg(P_{\mathcal{V}_{d_{k+1}}}\,C_{X;k-i}\,P_{\mathcal{V}_{d_{i+1}}} - \sum_{j=0}^{i-1} \theta_{k,k-j}  \ V_{j} \ \theta_{i,i-j}^*\bigg)V_{i}^{-1}, \qquad i=1,\dots,n-1, \label{theta1} \\
	V_{k}&=C_{X_{d_{k+1}}-\tilde{X}_{n+1,k}}= C_{X_{d_{k+1}}} - \sum_{i=0}^{k-1} \theta_{k,k-i}V_{i}\theta^*_{k,k-i}. \label{vd1}
	\end{align}
	Note that $\theta_{k,k-i}$ and $V_i$ are operators in $\call$ for all $i=1,\dots,k$.
\end{algo}

The first main goal is now to show how a finite sample version of this algorithm can be used to estimate the operators in \eqref{FMA}, as these FMA processes will be used to approximate the more complex processes appearing in Definition \ref{def:lpm}. Note that H\"ormann and Kokoszka \cite{weaklydep} give assumptions under which $\sqrt{n}$-consistent estimators can be obtained for the lag-$h$ autocovariance operator $C_{X;h}$, for $h\in\Z$. However, in \eqref{theta1}, estimators are required for the more complicated quantities $P_{\mathcal{V}_{d_{k+1}}}\,C_{X;k-i}\,P_{\mathcal{V}_{d_{i+1}}}$, for $k,\, i\in\N$. If, for $i\in\N$, the projection subspace $\mathcal{V}_{d_i}$ is known, consistent estimators of $P_{\mathcal{V}_{d_{k+1}}}\,C_{X;k-i}\,P_{\mathcal{V}_{d_{i+1}}}$ can be obtained by estimating $C_{X;k-i}$ and projecting the operator on the desired subspace. This case will be dealt with in Section~\ref{known}. In practice, however, the subspaces $\mathcal{V}_{d_i}$, $i\in\N$, need to be estimated. This is a further difficulty that will be addressed separately in an additional step as part of Section~\ref{unknown}.

Now, introduce additional notation. For $k\in\N$, denote by $(X_j(k)\colon j\in\Z)$ the functional process taking values in $H^k$ such that
\[
X_j(k)=(X_j,X_{j-1},\dots,X_{j-k+1})^\top,
\]
where $^\top$ signifies transposition. Let
\begin{align*}
\Gamma_k= C_{X(k)} \qquad \text{and} \qquad \Gamma_{1,k}=C_{X_{n+1},X_n(k)}=\E \big[X_{n+1} \otimes X_n(k)\big].
\end{align*}
Based on a realization $X_1,\dots,X_n$ of $(X_j\colon j\in\Z)$, estimators of the above operators are given by
\begin{align}
\hat\Gamma_k = \frac{1}{N-k}\sum_{j=k}^{N-1} X_j(k)\otimes X_j(k) 
\qquad\mbox{and}\qquad 
\hat\Gamma_{1,k} = \frac{1}{N-k} \sum_{j=k}^{N-1} X_{j+1}\otimes X_j(k).
\label{gammak}
\end{align}
%Note that if $(X_n)$ is a $L^4-m$-approximable functional linear process, so is $(X_n(k))_{n\in\Z}$. This is made more precise
The following theorem establishes the $\sqrt{n}$-consistency of the estimator $\hat\Gamma_k$ of $\Gamma_k$ defined in \eqref{gammak}.

\begin{theorem} \label{l4mapp}
If $(X_j\colon j\in\mathbb{Z})$ is a functional linear process defined in \eqref{eq:flp} such that the coefficient operators $(\psi_\ell\colon\ell\in\mathbb{N}_0)$ satisfy the summability condition $\sum_{m=1}^\infty\sum_{\ell=m}^\infty\|\psi_\ell\|_{\call}<\infty$ and with independent, identically distributed innovations $(\varepsilon_j\colon j\in\mathbb{Z})$ such that $\E[\|\varepsilon_0\|^4]<\infty$, then
%Let $(X_n)$ be a $L^4-m$-approximable functional linear process. Then $(X_n(k))_{n\in\Z}$ and $(P_{(d)}X_n(k))_{n\in\Z}$ are also  $L^4-m$-approximable functional linear processes. Furthermore
	\begin{align*}
	%(N-k)\,E\big[\Vert \hat C_{X(k)} - C_{X(k)}\Vert_\caln^2\big] =
	(N-k)\, \E \big[\Vert \hat \Gamma_k - \Gamma_k \Vert_\caln^2\big] 
	\leq k \, U_X,
	\end{align*}
	where $U_X$ 
	%= \E[\Vert X \Vert^4]+ 4\sqrt{2}(\E[\Vert X \Vert^4])^{3/4} \, \sum_{m=1}^{\infty} (\E[\Vert X_m- X_m^{(m)} \Vert^4])^{1/4}$ 
	is a constant that does not depend on $n$. %\textcolor{blue}{Result for $\hat\Gamma_{1,k}$?}
\end{theorem}

The proof of Theorem \ref{l4mapp} is given in Section \ref{sec:proof}. There, an explicit expression for the constant $U_X$ is derived that depends on moments of the underlying functional linear process and on the rate of decay of the temporal dependence implied by the summability condition on the coefficient operators $(\psi_\ell\colon\ell\in\mathbb{N}_0)$.

%The following result on the eigenvalues of $\Gamma_{k,d}$ will be very useful in what follows.
%\begin{lemma}[Theorem~1.2 in \cite{mitchell2}]\label{alphadk}
%	Let $(X_n)_{n\in\Z}$ be an invertible functional linear process. Then the smallest eigenvalue of $\Gamma_{k,d}$ is bounded below by $2\pi\,\alpha_{d_k}$, where $\alpha_{d_k}$ is  the infimum of the eigenvalues of all spectral density operators of $(X_{d_k,n})_{n\in\Z}$.
%\end{lemma}

%\JK{Hier kam vorher der Sch\"atzer von $V_k$ : den brauchen wir aus meiner Sicht allerdings gar nicht mehr, weil die gesch\"atzten $V_k$ aus dem Inno Algo hervor gehen..!}
%To complete the sample version of Algorithm \ref{fia}, estimators for the quantities $V_k$ are needed. Since these are the sample versions of the covariance operator of the one-step ahead prediction residual functions, this can be done using the estimator
%\begin{align}
%	\hat V_k = \frac{1}{N-k} \sum_{j=k}^{N-1} (X_j-\hat X_{j,k})\otimes (X_j-\hat X_{j,k}),
%\label{eq:vkhat}
%\end{align}
%where 
%\begin{align}
%$\hat X_{j,k} = \sum_{i=1}^{k} \hat \beta_{k,i} X_{j+1-i}$
%\end{align}
%is an estimate of $\tilde{X}_{j,k}$. The estimators in \eqref{eq:vkhat} are consistent.

%%%%%%%%%%%%%%%%%%%%%%%%
\subsection{Known projection subspaces}
\label{known}
%%%%%%%%%%%%%%%%%%%%%%%%

In this section, conditions are established that ensure consistency of estimators of a functional linear process under the assumption that the projection subspaces $\mathcal{V}_{d_i}$ are known in advance. In this case as well as in the unknown subspace case, the following the general strategy is pursued; see Mitchell and Brockwell \cite{mitchell}. Start by providing consistency results for the estimators regression estimators of $\beta_{k,1},\dots,\beta_{k,k}$ in the linear model formulation
\[
\tilde X_{n+1,k}=\beta_{k,1}X_{d_k,n} + \beta_{k,2}X_{d_{k-1},n-1}+\dots+\beta_{k,k}X_{d_1,n-k+1}
\]
of \eqref{blpd}. To obtain the consistency of the estimators $\theta_{k,1},\ldots,\theta_{k,k}$ exploit then that  
regression operators and Innovations Algorithm coefficient operators are, for $k\in\N$, linked through the recursions
\begin{align}
\theta_{k,i}=\sum_{j=1}^i \beta_{k,j} \theta_{k-j,i-j}, \qquad   i=1,\dots,k.\label{link}
\end{align}

Define furthermore $P_{(k)}= \mathrm{diag}(P_{\mathcal{V}_{d_k}}, \dots,P_{\mathcal{V}_{d_1}})$, the operator from $H^k$ to $H^k$ whose $i$th diagonal entry is given by the projection operator onto $\mathcal{V}_{d_i}$. One verifies that
%\begin{align}
$P_{(k)}X_n(k)= (X_{d_k,n},X_{d_{k-1},n-1},\dots,X_{d_1,n-k})^\top$, %\label{xnk }\\
$C_{P_{(k)}X(k)}=P_{(k)}\Gamma_k P_{(k)}=\Gamma_{k,d}$ and %\quad \text{and} \label{gammakd} \\ 
$C_{X,P_{(k)}X(k)}= P_{(k)}\Gamma_{1,k}=\Gamma_{1,k,d}$. %\label{gamma1kd}
%\end{align}
With this notation, it can be shown that $B(k)=(\beta_{k,1},\dots,\beta_{k,k})$ satisfies the population Yule--Walker equations
\begin{align*}
B(k)=\Gamma_{1,k,d}\,\Gamma_{k,d}^{-1},
\end{align*}
of which sample versions are needed. In the known subspace case, estimators of $\Gamma_{1,k,d}$ and $\Gamma_{k,d}$ are %simply 
given by %$\wh\Gamma_{1,k}P_{(k)}$ and $P_{(k)} \wh\Gamma_k P_{(k)}$,  To simplify the notation, we denote 
\begin{align} 
\wh\Gamma_{k,d}=P_{(k)} \wh\Gamma_k P_{(k)} \qquad\mbox{and}\qquad\wh\Gamma_{1,k,d}=\wh\Gamma_{1,k}P_{(k)}, \label{hatgammakd}
\end{align} 
where $\wh\Gamma_{k}$ and $\wh\Gamma_{1,k}$ are as in \eqref{gammak}. With this notation, $B(k)$ is estimated by the sample Yule--Walker equations
 \begin{align} 
  \wh{B}(k) = \wh\Gamma_{1,k,d}\wh\Gamma_{k,d} ^{-1}. \label{yulewalkerhat}
   \end{align}
Furthermore, the operators $\theta_{k,i}$ in \eqref{innov} are estimated by $\wh\theta_{k,i}$, resulting from Algorithm~\ref{fia} applied to the estimated covariance operators with $\mathcal{V}_{d_i}$ known. In order to derive asymptotic properties of $\wh \beta_{k,i}$ and $\wh \theta_{k,i}$ as both $k$ and $n$ tend to infinity, the following assumptions are imposed. Let $\alpha_{d_k}$ denote the infimum of the eigenvalues of all spectral density operators of $(X_{d_k,j}\colon j\in\mathbb{Z})$. 

\begin{assumption}\label{assumptions}
As $n\rightarrow\infty$, let $k=k_n\rightarrow\infty$ and $d_k\rightarrow\infty$ such that \vspace{-.2cm}
\begin{enumerate}\itemsep-.2ex
\item[(i)] $(X_j\colon j\in\Z)$ is as in Theorem \ref{l4mapp} and invertible. %a $L^4-m$-approximable, invertible functional linear process satisfying $\E\Vert X_n - X_n^{(l-i)}\Vert ^2 \leq \E\Vert X_n - X_n^{(l)}\Vert^2$ for  $i\geq 0$ %{Maybe add conditions concerning $k^{1/4}$.} %\JK{Is this really needed here? The only thing that we really need is that the estimators of the finite dimensional covariances converge.. Check conditions. It is a priori enough to have a independent noise process? be careful, as the finte dimensional process has a priori only uncorrelated noise (because of Wold Representation). Furthermore i think that if we have a linear process it should already be L2m approsximable}
\item[(ii)] $k^{1/2}(n-k)^{-1/2}\alpha_{d_k}^{-2}\rightarrow 0$ as  $n\rightarrow \infty$.
		%\item $k^{1/2} \,\alpha_{d_k}^{-1} (N-k)^{-1/2}\rightarrow 0$, 
\item[(iii)] $k^{1/2}\alpha_{d_k}^{-1}  \, \big(\sum_{\ell>k} \Vert \pi_\ell\Vert_{\call} + \sum_{\ell=1}^k \Vert \pi_\ell \Vert_{\call} \sum_{i>d_{k+1-\ell}} \lambda_i \big) \rightarrow 0$ as $n\rightarrow\infty$. 
	\end{enumerate}
\end{assumption}

Invertibility imposed in part {\em (i)}\/ of Assumption \ref{assumptions} is a standard requirement in the context of prediction and is also necessary for the univariate Innovations Algorithm to be consistent. %$L^4-m$-approximability is the dependence concept used to show convergence of the estimators of mean and covariance of the data. $\E\Vert X_n - X_n^{(l-i)}\Vert ^2 \leq \E\Vert X_n - X_n^{(l)}\Vert^2$ for  $i\geq 0$ is needed to show the $L^4-m$ approximability of the $H^k$ valued process $X_n(k)$. It is however not restrictive and is satisfied by every commonly known linear process. 
Assumption~{\em (ii)}\/ describes the restrictions on the relationship between $k$, $d_k$ and $n$. The corresponding multivariate assumption in Mitchell and Brockwell \cite{mitchell} is  $k^3/n\rightarrow 0 $ as $n\rightarrow\infty$. Assumption~{\em (iii)}\/ is already required in the population version of the functional Innovations Algorithm in Klepsch and Kl\"uppelberg \cite{kk}. It ensures that the best linear predictor based on the last $k$ observations converges to the conditional expectation for $k\rightarrow\infty$. The corresponding multivariate condition in Brockwell and Mitchell \cite{mitchell} is $k^{1/2} \sum_{\ell>k}\Vert \pi_\ell\Vert \rightarrow 0$ as $n\rightarrow\infty$, where $(\pi_\ell\colon\ell\in\mathbb{N})$ here denote the matrices in the invertible representation of a multivariate linear process. 

The main result concerning the asymptotic behavior of the estimators $\wh\beta_{k,i}$ and $\wh\theta_{k,i}$ is given next.

\begin{theorem}\label{autoregressive}
Let $\mathcal{V}_{d_i}$ be known for all $i\in\N$ and let Assumption~\ref{assumptions} be satisfied. Then, for all $x\in H$ and all $i\in\N$ as $n\rightarrow\infty$, \vspace{-.2cm}
\begin{enumerate}\itemsep-.2ex
\item[(i)] $\Vert (\wh \beta_{k,i} - \pi_i )(x) \Vert \overset{p}{\rightarrow} 0$,
\item[(ii)]$\Vert( \wh \theta_{k,i}-\psi_i)(x)\Vert \overset{p}{\rightarrow}0.$
\end{enumerate}
If the operators $(\psi_\ell\colon\ell\in\mathbb{N})$ and $(\pi_\ell\colon\ell\in\mathbb{N})$ in the respective causal and invertible representations are assumed Hilbert--Schmidt, then the convergence in (i) and (ii) is uniform.
\end{theorem}

The proof of Theorem \ref{autoregressive} is given in Section \ref{sec:proof}. The theorem establishes the pointwise convergence of the estimators needed in order to get a sample proxy for the functional linear filter $(\pi_\ell\colon\ell\in\N)$. This filter encodes the second-order dependence in the functional linear process and can therefore be used for estimating the underlying dynamics for the case of known projection subspaces.  %Convergence in the operator norm requires an additional assumption on the functional linear process. As is stated in e.g. Theorem~8.5 in \cite{bosq}, for uniform convergence we additionally need that $\pi_j$ and $\psi_j$ are Hilbert-Schmidt operators, for $j\in\N$.

%%%%%%%%%%%%%%%%%%%%%%%%
\subsection{Unknown projection subspaces}
\label{unknown}
%%%%%%%%%%%%%%%%%%%%%%%%

The goal of this section is to remove the assumption of known $\mathcal{V}_{d_i}$. Consequently, the standard estimators for the eigenfunctions $(\nu_i\colon i\in\N)$ of the covariance operator $C_X$ are used, obtained as the sample eigenfunctions $\wh\nu_j$ of $\wh C_X$. Therefore, for $i\in\N$, the estimators of $\mathcal{V}_{d_i}$ and $P_{\mathcal{V}_{d_i}}$ are
\begin{align}
\wh{\mathcal{V}}_{d_i}= \spa\{\wh \nu_1, \wh \nu_2,\dots,\wh \nu_{d_i}\} 
\qquad\mbox{and}\qquad 
\wh P_{\mathcal{V}_{d_i}} = P_{\wh{\mathcal{V}}_{d_i}}.
\end{align}
For $i\in\mathbb{N}$, let $\wh \nu_i'=c_i\wh \nu_i$, where $c_i = \text{sign}(\langle  \wh \nu_i,\nu_i\rangle)$. Then, Theorem 3.1 in H\"ormann and Kokoszka \cite{weaklydep} implies the consistency of $\wh\nu_i'$ for $\wh\nu_i$, with the quality of approximation depending on the spectral gaps of the eigenvalues $(\lambda_i\colon i\in\mathbb{N})$ of $C_X$.
%\begin{lemma}\label{lem:specgap}[\cite{weaklydep}, Theorem~3.1]
%	Let $(X_n)_{n\in\Z}$ be an $L^4-m$ approximable functional process. Then, for some constant $K<\infty$ for all $j\in\N$
%	\begin{align*}
%		N \E \Vert \wh \nu_j' - \nu_j \Vert ^2 \leq 8\,\delta_j^{-2} \, N \, \E \Vert \wh C_X - C_X \Vert_\caln ^2 \leq K \, \delta_j^{-2},
%	\end{align*}
%	where $\delta_1 = \lambda_1 - \lambda_2$ and $\delta_j = \min(\lambda_{j-1}-\lambda_j,\lambda_j-\lambda_{j+1})$, for $j\geq 2$.
%\end{lemma}
With this result in mind, define 
\begin{align}
\wh	{\wh\Gamma}_{k,d}=\wh P_{(k)} \wh\Gamma_k \wh P_{(k)} 
\qquad\mbox{and}\qquad
\wh{\wh\Gamma}_{1,k,d}=\wh\Gamma_{1,k}\wh P_{(k)}. 
\label{hathatgammakd}
\end{align}
Now, if the projection subspace $\mathcal{V}_{d_i}$ is not known, the operators appearing in \eqref{link} and can be estimated by solving the estimated Yule--Walker equations 
\begin{align} 
 \wh{\wh{B}}(k) = \wh{\wh\Gamma}_{1,k,d}\wh{\wh\Gamma}_{k,d} ^{-1}. \label{yulewalkerhathat}
 \end{align}
The coefficient operators in Algorithm \ref{fia} obtained from estimated covariance operators and estimated projection space $\wh P_{\mathcal{V}_{d_i}} $ are denoted by $\wh{\wh\theta}_{k,i}$. In order to derive results concerning their asymptotic behavior, an additional assumption concerning the decay of the spectral gaps of $C_X$ is needed. Let $\delta_1=\lambda_1-\lambda_2$ and $\delta_j=\min\{\lambda_{j-1}-\lambda_j,\lambda_j-\lambda_{j+1}\}$ for $j\geq 2$.

\begin{assumption} \label{ass2}
As $n\rightarrow\infty$, $k=k_n\rightarrow\infty$ and $d_k\rightarrow\infty$ such that
\begin{enumerate}
\item[(iv)] $k^{3/2}{\alpha_{d_k}^{-2} \, n^{-1}}(\sum_{\ell=1}^{d_k} \delta_\ell^{-2})^{1/2} \rightarrow 0$.
\end{enumerate}
\end{assumption}

This type of assumption dealing with the spectral gaps is typically encountered when dealing with the estimation of eigenelements of functional linear processes (see, for example, Bosq \cite{bosq}, Theorem~8.7). We are now ready to derive the asymptotic result of the estimators in the general case that $A_{d_i}$ is not known.

\begin{theorem}\label{autoregressive2}
Let Assumptions~\ref{assumptions} and \ref{ass2} be satisfied. Then, for all $x\in H$ and $i \in\N$ as $n\rightarrow\infty$, 
\vspace{-.2cm}
\begin{enumerate}\itemsep-.2ex
\item[(i)]  $\Vert (\wh{\wh \beta}_{k,i} - \pi_i )(x) \Vert \overset{p}{\rightarrow} 0,$
\item[(ii)] $	\Vert( \wh{\wh\theta}_{k,i}-\psi_i)(x)\Vert \overset{p}{\rightarrow}0.$
\end{enumerate}
If the operators $(\psi_\ell\colon\ell\in\mathbb{N})$ and $(\pi_\ell\colon\ell\in\mathbb{N})$ are Hilbert--Schmidt, then the convergence %in {\em (i)}\/ and {\em (ii)}\/
is uniform.
\end{theorem}

The proof of Theorem \ref{autoregressive2} is given in Section \ref{sec:proof}. The theoretical results quantify the large-sample behavior of the estimates of the linear filter operators in the causal and invertible representations of the strictly stationary functional time series $(X_j\colon j\in\mathbb{Z})$. How to guide the application of the proposed method in finite samples is addressed in the next section.

%%%%%%%%%%%%%%%%%%%%%%%%
\section{Selection of principal directions and FMA order}
\label{sec:selection}
%%%%%%%%%%%%%%%%%%%%%%%%

Model selection is a difficult problem when working with functional time series. Contributions to the literature have been made in the context of functional autoregressive models by Kokoszka and Reimherr \cite{kokoreim}, who devised a sequential test to decide on the FAR order, %Fink \cite{fink}, who developed a procedure for order selection based on an AICC-type criterion, 
and Aue et al.\ \cite{aue}, who introduced an FPE-type criterion. %The performance of these criteria depends on the goal of the statistical analysis. 
% 
%The decision on which method to use should be left to the practitioner: for instance, when the goal is of a purely predictive nature, it is difficult to beat a FPE-type criterion as proposed in \cite{aue}. When the goal is estimation,  model selection, or minimizing the variance or dependence in the residuals, other criteria might have a better performance. 
%
%Even though for theoretical considerations, we assumed to project our data on spaces of dimension increasing with $k$ to achieve consistency, we will for practical considerations take $d$ fixed. Choosing $d$ is a crucial and non-trivial task. Choosing $d$ to small might induce a loss of relevant information on the dependence structure of the data. This is especially relevant as the basis chosen by fPCA for dimension reduction might not optimal in terms of capturing the dependence structure. On the other hand, choosing $d$ to large leads to very high dimensional vector models where estimators have potentially high variances. We stress again that the choice of $d$ is highly dependent on the goal of the practitioner. 
%
To the best of our knowledge, there are no contributions in the context of model selection in functional moving average models. This section introduces several procedures. %We therefore propose different options for the selection of $d$ and $q$, where the choice of which method to use depends on the goal of the practitioner. Deriving consistency results of the methods goes beyond the scope of the paper, but we conduct simulation studies to evaluate the performance. 
%There are two different approaches in the selection of $d$ and $q$: one can either choose $d$ and $q$ simultaneously or start by selecting $d$ and then perform the selection of $q$ on the resulting $d$-dimensional vector process. 
A method for the selection of the subspace dimension is introduced in Section~\ref{sec:cvpind}, followed by a method for the FMA order selection in Section~\ref{sec:AICC}. A criterion for the simultaneous selection is in Section~\ref{sec:fFPE}.

\subsection{Selection of principal directions}%Total variance explained and a test of independence} 
\label{sec:cvpind}

The most well-known method for the selection of $d$ in functional data analysis is based on total variance explained, TVE, where $d$ is chosen such that the first $d$ eigenfunctions of the covariance operator explain a predetermined amount $P$ of the variability; see, for example, Horv\'ath and Kokoszka \cite{horvath}. %The method is shown to be optimal in the context of i.i.d. data (e.g. \cite{horvath}) in terms of representation of the variability. However, as is stated in  \cite{hoermann}, the choice of $d$ is just shifted to the choice of $P$, the 'correct' amount of variability to be explained. %Furthermore, in a setting with temporal dependence in the data, other basis functions might be more suited for dimension reduction (see \cite{hoermann}). 
%The goal of a dimension reduction in this context is to capture as much of the dependence in the data as possible, which is not guaranteed by a CPV type method. 
In order to apply the TVE criterion in the functional time series context, one has to ensure that no essential parts of the dependence structure in the data are omitted after the projection into principal directions. This is achieved as follows. First choose an initial $d^*$ with the TVE criterion such with a fraction $P$ of variation in the data is explained. This should be done conservatively. Then apply the portmanteau test of Gabrys and Kokoszka \cite{gabrys} to check whether the non-projected part $(I_H-P_{\mathcal{V}_{d^*}})X_1,\ldots, (I_H-P_{\mathcal{V}_{d^*}})X_n$ of the observed functions $X_1,\ldots,X_n$ can be considered independent. Modifying their test to the current situation, yields the statistic
\begin{align}\label{testind}
	Q_n^{d^*}= n\sum_{h=1}^{\bar{h}} \sum_{\ell,\ell^\prime=d^*+1}^{d^*+p} f_{h}(\ell,\ell^\prime) b_h(\ell,\ell^\prime),
\end{align}
where $f_h(\ell,\ell^\prime)$ and $b_h(\ell,\ell^\prime)$ denote the $(\ell,\ell^\prime)$th entries of $C_{\mathbf{X}^*;0}^{-1}C_{\mathbf{X}^*;h}$ and $C_{\mathbf{X}^*;h}C_{\mathbf{X}^*;0}^{-1}$, respectively, and $(\textbf{X}_j^*\colon j\in\Z)$ is the $p$-dimensional vector process consisting of the $d+1$st to $d+p$th eigendirections of the covariance operator $C_X$. Following Gabrys and Kokoszka \cite{gabrys}, it follows under the assumption of independence of the non-projected series that $Q_N^{d^*}\rightarrow\chi^2_{{p}^2\bar{h}}$ in distribution. If the assumption of independence is rejected, set $d^*=d^*+1$. Repeat the test until the independence hypothesis cannot be rejected and choose $d=d^*$ to estimate the functional linear filters. This leads to the following algorithm.

\begin{algo}[\bf Test for independence]
\label{IND}
Perform the following steps.
\vspace{-.2cm}
\begin{enumerate}\itemsep-.2ex
\item[(1)] For given observed functional time series data $X_1,\ldots,X_n$, estimate the eigenpairs $(\wh\lambda_1,\wh\nu_1),\dots,(\wh\lambda_n,\wh\nu_n)$ of the covariance operator $C_X$. Select $d^*$ such that
\begin{align*}
\mathrm{TVE}(d^*)=\frac{\sum_{i=1}^{d^*} \wh\lambda_i}{\sum_{i=1}^{n} \wh\lambda_i}\geq P
\end{align*}
for some prespecified $P\in(0,1)$.
\item[(2)] While $Q_n^{d^*}>q_{\chi^2_{{p}^2\bar{h}},\alpha}$, set $d^*=d^*+1$.
\item[(3)] If $Q_n^{d^*}\leq q_{\chi^2_{{p}^2\bar{h}},\alpha}$ stop and apply Algorithm \ref{fia} with $d_i=d^*$, for all $i\leq k$. 
\end{enumerate}
\end{algo} 

Note that the Algorithm \ref{IND} does not specify the choices of $P$, $p$, $H$ and $\alpha$. Recommendations on their selection are given in Section~\ref{sec:sim}. Multiple testing could potentially be an issue, but intensive simulation studies have shown that, since $d^*$ is initialized with the TVE criterion, usually no more than one or two iterations and tests are required for practical purposes. Therefore the confidence level is not adjusted, even though it would be feasible to incorporate this additional step into the algorithm.

\subsection{Selection of FMA order}%Choice of $q$ based on multivariate process}
\label{sec:AICC}

For a fixed $d$, multivariate model selection procedures can be applied to choose $q$. In fact, it is shown in Theorem~4.7 of Klepsch and Kl\"uppelberg \cite{kk} that the projection of an FMA$(q)$ process on a finite-dimensional space is a VMA$(q^*)$ with $q^*\leq q$. Assuming that the finite-dimensional space is chosen such that no information on the dependence structure of the process is lost, $q=q^*$. Then, the FMA order $q$ may be chosen by performing model selection on the $d$-dimensional vector model given by the first $d$ principal directions of $(X_j\colon j\in\Z)$. Methods for selecting the order of VMA models are described, for example, in Chapter 11.5 of Brockwell and Davis \cite{brockwell}, and Chapter~3.2 of Tsai \cite{tsai}. 

The latter book provides arguments for the identification of the VMA order via cross correlation matrices. This Ljung--Box (LB) method for testing the null hypothesis $H_0\colon C_{\textbf{X};\underline{h}} = C_{\textbf{X};\underline{h}+1 } = \dots =  C_{\textbf{X};\overline{h}} = 0$ versus the alternative that $C_{\textbf{X};h}\neq 0$ for a lag $h$ between $\underline{h}$ and $\overline{h}$ is based on the statistic  
\begin{align} \label{ljungbox}
	Q_{\underline{h},\overline{h}}
	= n^2 \sum_{h=\underline{h}}^{\overline{h}} \frac{1}{n-h} 
	\mathrm{tr} ( \wh C_{\textbf{X};h}^\top \wh C_{\textbf{X};0}^{-1} \wh C_{\textbf{X};h}^{\phantom{-1}}C_{\textbf{X};0 } ^{-1}).
\end{align}
Under regularity conditions $Q_{\underline{h},\overline{h}}$ is asymptotically distributed as a $\chi^2 _{d ^2(\overline{h}-\underline{h}+1)}$ random variable if the multivariate procss $(\textbf{X}_j\colon j\in\Z)$ on the first $d$ principal directions follows a VMA$(q)$ model and $\underline{h}>q$. For practical implementation, one computes iteratively $Q_{1,\overline{h}}, Q_{2,\overline{h}},\ldots$ and selects the order $q$ as the largest $\underline{h}$ such that $Q_{\underline{h},\overline{h}}$ is significant, but $Q_{\underline{h}+h,\overline{h}}$ is insignificant for all $h>0$. 

Alternatively, the well-known AICC criterion could be utilized. Algorithm \ref{fia} allows for the computationally efficient maximization of the likelihood function through the use of its innovation form; see Chapter 11.5 of Brockwell and Davis~\cite{brockwell}.
%\begin{align*}%\label{likelihood}
%L (\gamma_1,\dots,\gamma_q,\Sigma) = (2\pi)^{-Nd/2} \, \Big( \prod_{j=1}^N \det {\bf V }_{j-1} \Big) ^{-1/2} \exp\big(-\frac{1}{2} \sum_{j=1}^{N} ({\bf X}_j - {\wh{\bf X}}_j) {\bf V}_{j-1} ({\bf X}_j - \wh{\bf X}_j)^*\big),
%\end{align*}
%where ${\bf V}_j$ and $\wh{\bf X}_j$ are given by the Innovations algorithm in \eqref{xdhat} and \eqref{vd}. 
The AICC criterion is then given by
\begin{align}\label{AICC}
	\mathrm{AICC}(q)=-2 \ln L (\Theta_1,\dots,\Theta_q,\Sigma) + \frac{2 n d ( q  d^2 +1 )}{n d - qd^2 -2 },
\end{align}
where $\Theta_1,\ldots,\Theta_q$ are the fitted VMA coefficient matrices and $\Sigma$ its fitted covariance matrix. The minimizer of \eqref{AICC} is selected as order of the FMA process. Both methods are compared in Section~\ref{sec:sim}.

\subsection{Functional FPE criterion} \label{sec:fFPE}

In this section a criterion that allows to choose $d$ and $q$ simultaneously is introduced. A similar criterion was established in Aue et al.\ \cite{aue}, based on a decomposition of the functional mean squared prediction error. Note that, due to the orthogonality of the eigenfunctions $(\nu_i\colon i\in\N)$ and the fact that $\wh X_{n+1,k}$ lives in $\mathcal{V}_d$, 
\begin{align} \label{decomposition}
\E\big[ \Vert X_{n+1} - \wh X_{n+1,k} \Vert ^ 2\big] 
&= \E \big[\Vert P_{\mathcal{V}_d} (X_{n+1} - \wh X_{n+1,k}) \Vert ^2 \big]+ \E\big[ \Vert (I_H-P_{\mathcal{V}_d}) X_{n+1} \Vert ^2\big].
\end{align}
The second summand in \eqref{decomposition} satisfies $\E[ \Vert (I_H-P_{\mathcal{V}_d}) X_{n+1} \Vert ^2] = \E[ \Vert \sum_{i>d} \langle X_{n+1} , \nu_i \rangle \nu_i \Vert ^2] = \sum_{i>d} \lambda_i$. The first summand in \eqref{decomposition} is, due to the isometric isomorphy between $\mathcal{V}_d$ and $\R^d$ equal to the mean squared prediction error of the vector model fit on the $d$ dimensional principal subspace. It can be shown using the results of Lai and Lee \cite{lai} that it is of order $\mathrm{tr} (C_\mathbf{Z}) + qd\,\mathrm{tr}(C_\mathbf{Z})/n$, where $C_\textbf{Z}$ denotes the covariance matrix of the innovations of the vector process. Using the matrix version $\mathbf{V}_n$ of the operator $V_n$ given through Algorithm~\ref{fia} as a consistent estimator for $C_{\textbf{Z}}$, the functional FPE criterion
\begin{align}
	\text{fFPE}(d,q) = \frac{n + q \,d}{n}\, \mathrm{tr}(\mathbf{V}_n)+\sum_{i>d}\hat\lambda_i \label{fFPE}
\end{align}
is obtained. It can be minimized over both $d$ and $q$ to select the dimension of the principal subspace and the order of the FMA process jointly. As is noted in Aue et al.\ \cite{aue}, where a similar criterion is proposed for the selection of the order of an FAR$(p)$ model, the fFPE method is fully data driven: no further selection of tuning parameters is required. %Furthermore, note the difference between \eqref{fFPE} to the original FPE method introduced in a multivariate context in \cite{akaike}: the fFPE criterion uses the trace instead of the log-determinant. This makes sure that the two terms in \eqref{fFPE} are on the same scale (see \cite{aue}).

%%%%%%%%%%%%%%%%%%%%%%%%
\section{Simulation evidence}
\label{sec:sim}
%%%%%%%%%%%%%%%%%%%%%%%%

%%%%%%%%%%%%%%%%%%%%%%%%
\subsection{Simulation setting}
\label{subsec:sim:setting}
%%%%%%%%%%%%%%%%%%%%%%%%

In this section, results from Monte Carlo simulations are reported. The simulation setting was as follows. Using the first $D$ Fourier basis functions $f_1,\ldots,f_D$, the $D$-dimensional subspace $G^D=\spa\{f_1,\ldots,f_D\}$ of $H$ was generated following the setup in Aue et al.\ \cite{aue}, then the isometric isomorphy between $\R^D$ and $G^D$ is utilized to represent elements in $G^D$ by $D$-dimensional vectors and operators acting on $G^D$ by $D\times D$ matrices. Therefore $N+q$ $D$-dimensional random vectors as innovations for an FMA$(q)$ model and $q$ $D\times D$ matrices as operators were generated. Two different settings were of interest: processes possessing covariance operators with slowly and quickly decaying eigenvalues. Those cases were represented by  selecting two sets of standard deviations for the innovation process, namely
\begin{align}
\sigma_{\text{slow}} = (i^{-1}\colon i=1,\dots,D) 
\qquad\mbox{and}\qquad
\sigma_{\text{fast}} = (2^{-i} \colon i=1,\dots,D).	
\end{align}
With this, innovations 
\begin{align*}
\varepsilon_j= \sum_{i=1}^{D} c_{j,i} f_i,
\qquad j=1-q,\ldots,n,
\end{align*}
were simulated, where $c_{j,i}$ are independent normal random variables with mean $0$ and standard deviation $\sigma_{\cdot,i}$, the $\cdot$ being replaced by either slow or fast, depending on the setting. The parameter operators $\tilde\theta_\ell$, for $\ell=1,\dots,q$, were chosen at random by generating  $D\times D$ matrices, whose entries $\langle \tilde\theta_\ell f_i , f_{i'} \rangle $ were independent zero mean normal random variables with variance $\sigma_{\cdot,i}\sigma_{\cdot,i'}$. The matrices were then rescaled to have spectral norm $1$. Combining the forgoing, the FMA($q$) process  
\begin{align} \label{simfma}
	X_j = \sum_{\ell=1}^q\theta_\ell \varepsilon_{j-\ell} + \varepsilon_j,
	\qquad j=1,\ldots,n
\end{align} 
were simulated, where $\theta_\ell=\kappa_\ell \tilde\theta_\ell$ with $\kappa_\ell$ being chosen to ensure invertibility of the FMA process. In the following section, the performance of the proposed estimator is evaluated, and compared and contrasted to other methods available in the literature for the special case of FMA(1) processes, in a variety of situations.

%%%%%%%%%%%%%%%%%%%%%%%%
\subsection{Estimation of FMA(1) processes}\label{sec:fma1}
%%%%%%%%%%%%%%%%%%%%%%%%

In this section, the performance of the proposed method is compared to two approaches introduced in Turbillon et al.\ \cite{turbillon2} for the special case of FMA(1) processes. These methods are based on the following idea.
%\begin{align*}
%	X_n = \kappa_1 \gamma_1 \varepsilon_{n-1} + \varepsilon_n, \quad n\in \Z,
%\end{align*}
%where $(\varepsilon_n)_{n\in\Z}$ is white noise and has finite fourth moments and $\gamma_1 \in \call$. In order to be able to evaluate the performance of our method, we compare it with two estimators of FMA$(1)$ models, developed by Bosq and Turbillon in \cite{turbillon}. We start by introducing those two estimators.
%
Denote by $C_\varepsilon$ the covariance operator of $(\varepsilon_n\colon n\in\Z)$.
Observe that since $C_{X;1}= \theta_1 C_\varepsilon $ and $C_{X} = C_\varepsilon + \theta_1 C_\varepsilon \theta_1^*$, it follows that $\theta_1 C_X  = \theta_1 C_\varepsilon+\theta_1^2C_\varepsilon \theta_1 ^\ast=C_{X;1}+\theta_1^2C_{X;1}^\ast$,
and especially
\begin{align}
\theta_1 ^2 C_{X;1}^\ast-\theta_1 C_X+ C_{X;1}=0.
\label{lquad}
\end{align}
The estimators in Turbillon et al.\ \cite{turbillon2} are based on solving the quadratic equation in \eqref{lquad} for $\theta_1$. The first of these only works under the restrictive assumption that $\theta_1$ and $C_\varepsilon$ commute. Then, solving \eqref{lquad} is equivalent to solving univariate equations generated by individually projecting \eqref{lquad} onto the eigenfunctions of $C_X$. %Solving the univariate quadratic equations for the diagonal entries of $\wh \gamma_1$ leads to a consistent estimator of $\gamma_1$ under assumptions specified in \cite{turbillon2}. However the assumption that $\gamma_1$ and $C_\varepsilon$ commute is very restrictive. Furthermore, as we will see in the upcoming results, a breach of this condition leads to poor performance of the method.
The second approach is inspired by the Riesz--Nagy method.  It relies on regarding (\ref{lquad}) as a fixed-point equation and therefore establishing a fixed-point iteration. %Suppose we want to find a symmetric operator $A$ satisfying $A^2 R - A + R=0$ holds for some positive definite and symmetric operator $R$ with $\Vert R \Vert_\call \leq 1/2$.  Observe that the recursion
%\[
%A_0:=0,\qquad A_{r+1}:=\frac12(R+A^2_r),\qquad\forall r\geq0,
%\]
%can be established, and by \cite{vigier} and Proposition~2.1 in \cite{turbillon}, $A_r$ converges to the unique solution of  $A^2 R - A + R=0$. To use the result to solve \eqref{lquad} one could think of multiplying \eqref{lquad} by $C_X^{-1}$. As $C_X$ is generally not invertible on $H$, one can project the equation on $A_d$ as defined in \eqref{xdi}. Here $P_{A_d} C_X P_{A_d}$ is invertible on $A_d$ and the above recursion can be used to solve \eqref{lquad} on $A_d$. 
Since solutions may not exist in $H$, suitable projections have to be applied. Consistency of both estimators is established in Turbillon et al.\ \cite{turbillon2}. 

%We do not give detailed information on those methods, but note that under technical assumptions consistency of both methods is shown in \cite{turbillon}. Furthermore we note that similarly to our method, for practical implementation both methods require the projection of the data on some finite dimensional subspace $A_d$ of $H$.

To compare the performance of the methods, FMA$(1)$ time series were simulated as described in Section~\ref{subsec:sim:setting}. As measure of comparison the estimation error $\Vert \theta_1 - \wh\theta_1\Vert_\call$ was used after computing $\wh\theta_1$ with the three competing procedures. Rather than selecting the dimension of the subspace via Algorithm \ref{IND}, the estimation error is computed for $d=1,\ldots,5$. The results are summarized in Table~\ref{esterror0.8}, where estimation errors were averaged over 1000 repetitions for each specification, using sample sizes $n=100,500$ and $1{,}000$. 
\begin{table}[ht]
\vspace{.3cm}
\centering
\begin{tabular}{rrrrrrrrrrr}
	\hline
	&\ & \multicolumn{3}{c}{\textbf{$n=100$}} &  \multicolumn{3}{c}{\textbf{$n=500$}} \ & \multicolumn{3}{c}{\textbf{$n=1000$}} \\
	&$d$	& Proj &  Iter & Inn  &Proj   & Iter & Inn & Proj &  Iter & Inn \\ 
	\hline
	\multirow{5}{*}{$\sigma_\text{fast}$} &
	1 & 0.539 & 0.530 & 0.514 & 0.527 & 0.521 & 0.513 & 0.518 & 0.513 & 	0.508 \\
	&2 & 0.528 & 0.433 & \bf 0.355 & 0.508 & 0.391 & 0.287 & 0.500 & 0.386 & 0.277 \\
	&3 & 0.533 & 0.534 & 0.448 & 0.512 & 0.467 & \bf 0.235 & 0.503 & 0.460 & 	\bf 0.197 \\
	&4 & 0.534 & 0.650 & 0.582 & 0.513 & 0.573 & 0.276 & 0.504 & 0.567 & 	0.216 \\
	&5 & 0.534 & 0.736 & 0.646 & 0.513 & 0.673 & 0.311 & 0.504 & 0.662 & 	0.239 \\
	%&6 & 0.534 & 0.782 & 0.663 & 0.513 & 0.722 & 0.322 & 0.504 & 0.711 & 0.246 \\
	%7 & 0.534 & 0.798 & 0.668 & 0.513 & 0.742 & 0.325 & 0.504 & 0.731 & 0.248 \\
	%8 & 0.534 & 0.803 & 0.669 & 0.513 & 0.748 & 0.325 & 0.504 & 0.735 & %0.249 \\
	\hline
	\multirow{6}{*}{$\sigma_\text{slow}$} &
	1 & 0.610 & 0.602 & 0.588 & 0.579 & 0.574 & 0.566 & 0.575 & 0.573 & 0.569 \\
	&2 & 0.614 & 0.527 & \bf 0.513 & 0.581 & 0.487 & 0.434 & 0.577 & 0.483 & 0.422 \\
	&3 & 0.618 & 0.552 & 0.610 & 0.583 & 0.504 & \bf 0.389 & 0.578 & 0.500 & 0.362 \\
	&4 & 0.620 & 0.591 & 0.861 & 0.584 & 0.531 & 0.402 & 0.579 & 0.522 & \bf 0.344 \\
	&5 & 0.620 & 0.630 & 1.277 & 0.584 & 0.556 & 0.448 & 0.579 & 0.548 & 0.358 \\
	%&6 & 0.621 & 0.669 & 1.985 & 0.584 & 0.583 & 0.519 & 0.579 & 0.570 & 0.393 \\
	%7 & 0.621 & 0.703 & 3.343 & 0.584 & 0.610 & 0.603 & 0.579 & 0.592 & %0.441 \\
	%8 & 0.621 & 0.729 & 6.473 & 0.584 & 0.634 & 0.706 & 0.579 & 0.612 & %0.502 \\
	\hline
\end{tabular}
\label{esterror0.8}
\caption{Estimation error $\Vert \theta_1 -{\wh{\theta}}_1 \Vert_\call$, with $\theta_1=\kappa_1\tilde\theta_1$ and $\kappa_1 = 0.8$, with $\wh{\theta}_1$ computed with the projection method (Proj) and the iterative method (Iter) of \cite{turbillon2}, and the proposed method based on the functional Innovations Algorithm (Inn). The smallest estimation error is highlighted in bold for each case.}
\end{table}

For all three sample sizes, the operator kernel estimated with the proposed algorithm is closest to the real kernel. As can be expected, the optimal dimension increases with the sample size, especially for the case where the eigenvalues decay slowly. The projection method does not perform well, which is also to be expected, because the condition of commuting $\theta_1$ and $C_\varepsilon$ is violated. One can see that the choice of $d$ is crucial: especially for small sample sizes for the proposed method, the estimation error explodes for large $d$. In order to get an intuition for the shape of the estimators, the kernels of the estimators resulting from the different estimation methods, using $n=500$ and $\kappa_1=0.8$, are plotted in Figure~\ref{kernelslow0.8}. It can again be seen that the projection method yields results that are significantly different from both the truth and the other two methods who produce estimated operator kernels, whose shapes look roughly similar to the truth. 
\begin{figure}[ht]
	\begin{center}
		% Requires \usepackage{graphicx}[height=7in,width=6.25in]
		\includegraphics[width=\linewidth]{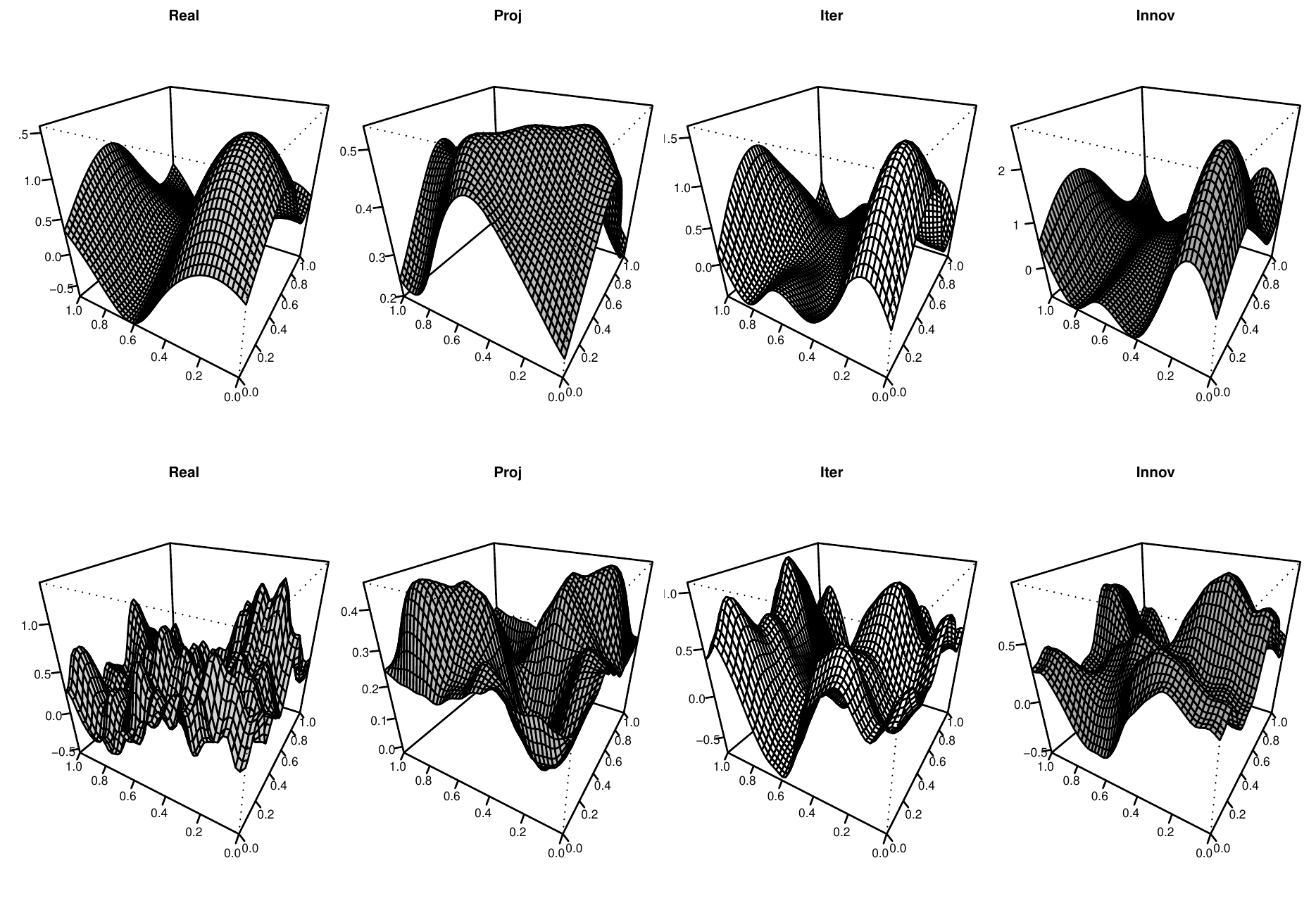}\\
		\caption{Estimated operator kernel of simulated FMA$(1)$ process with $\kappa_1=0.8$, $d=3$ and $\sigma_\text{fast}$  (first row) and $\sigma_{\text{slow}}$ (second row), using $n=500$ sampled functions. Labeling of procedures is as in Table \ref{esterror0.8}.
		\label{kernelslow0.8}}
	\end{center}
\end{figure}

%%%%%%%%%%%%%%%%%%%%%%%%
\subsection{Model selection}
%%%%%%%%%%%%%%%%%%%%%%%%

In this section, the performance of the different model selection methods introduced in Section~\ref{sec:selection} is demonstrated. To do so, FMA(1) processes with weights $\kappa_1=0.4$ and $0.8$ were simulated as in the previous section. In addition, two different FMA$(3)$ processes were simulated according to the setting described in Section \ref{subsec:sim:setting}, namely
%\begin{align*}
% X_n =\kappa_1 \gamma_1 \varepsilon_{n-1} + \kappa_2\gamma_2 \varepsilon_{n-2} +\kappa_3\gamma_3 \varepsilon_{n-3} + \varepsilon_n, \quad n\in \Z.
%\end{align*}
\vspace{-.1cm}
\begin{itemize}\itemsep-.2ex
	\item Model 1: $\kappa_1=0.8$, $\kappa_2=0.6$, and $\kappa_3=0.4$. 
	\item Model 2: $\kappa_1=0$, $\kappa_2=0$, and $\kappa_3=0.8$.
\end{itemize}
For sample sizes $n=100$, $500$ and $1{,}000$, $1{,}000$ processes of both Model 1 and 2 were simulated using  $\sigma_\text{slow}$ and $\sigma_\text{fast}$. The estimation process was done as follows. First, the dimension $d$ of the principal projection subspace was chosen using Algorithm \ref{IND} with TVE such that $P=0.8$. With this selection of $d$, the LB and AICC criteria described in Section~\ref{sec:AICC} were applied to choose $q$. Second, the fFPE criterion was used for a simultaneous selection of $d$ and $q$.  The results are summarized in Figures~\ref{boxplotma1} and \ref{boxplotma3}.

\begin{figure}[ht]
	\begin{center}
		% Requires \usepackage{graphicx}[height=7in,width=6.25in]
		\includegraphics[width=\linewidth]{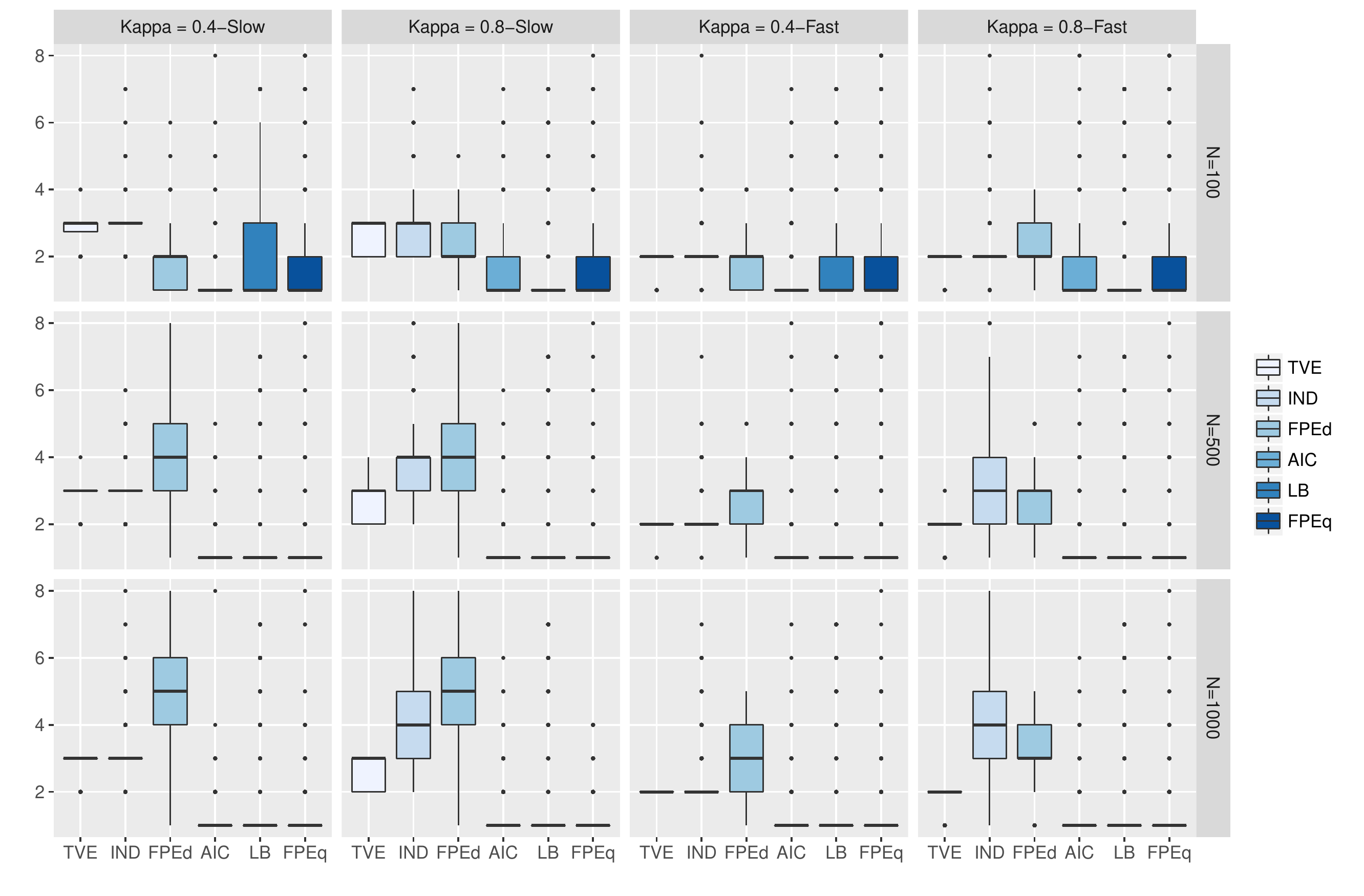}\\
		\caption{Model selection for different MA(1) processes. The left three plots in each small figure give the $d$ chosen by total variation explained with $P=0.8$ (TVE), Algorithm \ref{IND} (IND) and the functional FPE criterion (FPEd). The right three plots in each small figure give the selected order $q$ by AICC, LB and fFPE.
		\label{boxplotma1}}
	\end{center}
\end{figure}
\begin{figure}[ht]
	\begin{center}
		% Requires \usepackage{graphicx}[height=7in,width=6.25in]
		\includegraphics[width=\linewidth]{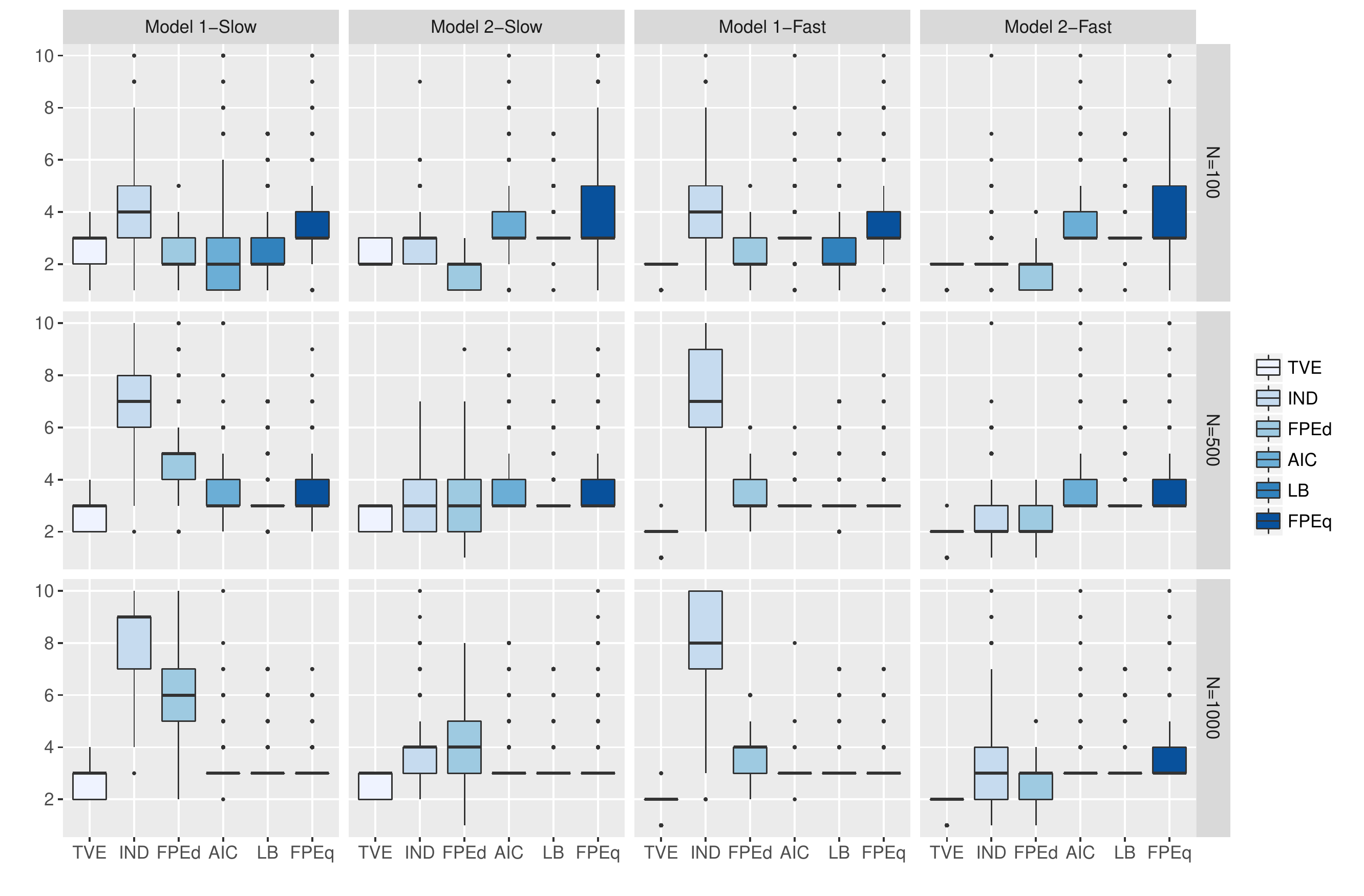}\\
		\caption{Model selection for different MA(3) processes. Labeling of procedures is as in Figure~\ref{boxplotma1}.}
		\label{boxplotma3}
	\end{center}
\end{figure}

Figures~\ref{boxplotma1} and \ref{boxplotma3} allow for a number of interesting observations. For both the FMA$(1)$ and the FMA$(3)$ example, the model order is estimated well. In all cases, especially for sample sizes larger than 100, all three selection methods (AIC, LB, FPEq) for the choice of $q$ yield the correct model order (1 or 3). The Ljung--Box (LB) method seems to have the most stable results. The methods for the choice of $d$ are more heterogeneous. The TVE method yields the most stable results among different sample sizes. For $\sigma_\text{fast}$, it almost always selects $d=2$ and for $\sigma_\text{slow}$ the choice varies between $d=2$ and $d=3$. However, the TVE method seems to underestimate $d$. Often there appears to be dependence left in the data, as one can see from the selection of $d$ by Algorithm \ref{IND}. Especially in the FMA$(3)$ case and Model~1, this algorithm yields some large choices for $d$ of about $7$ or $8$. The choice of FPEd seems to increase with increasing sample size: this is to be expected as for increasing sample size the variance of the estimators decreases and the resulting predictors get more precise, even for high-dimensional models. This is valid especially for $\sigma_\text{slow}$ where a larger $d$ is needed to explain the dynamics of the functional process. A similar trade-off is occasionally observed for Algorithm \ref{IND}.

%%%%%%%%%%%%%%%%%%%%%%%%
\section{Application to traffic data}
\label{sec:app}
%%%%%%%%%%%%%%%%%%%%%%%%

In this section, the proposed estimation method is applied to vehicle traffic data provided by the Autobahndirektion S\"udbayern. The dataset consists of measurements at a fixed point on a highway (A92) in Southern Bavaria, Germany. Recorded is the average velocity per minute from 1/1/2014 00:00 to 30/06/2014 23:59 on three lanes. After taking care of missing values and outliers, the velocity per minute was averaged over the three lanes, weighted by the number of vehicles per lane. This leads to $1440$ preprocessed and cleaned data points per day, which were transformed into functional data using the first $30$ Fourier basis functions with the \texttt{R} package \texttt{fda}. The result is a functional time series $(X_j\colon j=1,\ldots,n=119)$, which is deemed stationary and exhibits temporal dependence, as evidenced by Klepsch et al.\ \cite{KKW}.

The goal then is to approximate the temporal dynamics in this stationary functional time series with an FMA fit. Observe that the plots of the spectral norms $\Vert \wh C_{\mathbf{X};h}\wh C_{\mathbf{X};0}^{-1}\Vert_\call$ for $h=0,\dots,5$ in Figure \ref{acf} display a pattern typical for MA models of low order. Here $\mathbf{X}$ stands for the multivariate auxiliary model of dimension $d$ obtained from projection into the corresponding principal subspace.
\begin{figure}[ht]
	\begin{center}
		% Requires \usepackage{graphicx}[height=7in,width=6.25in]
		\includegraphics[width=\linewidth]{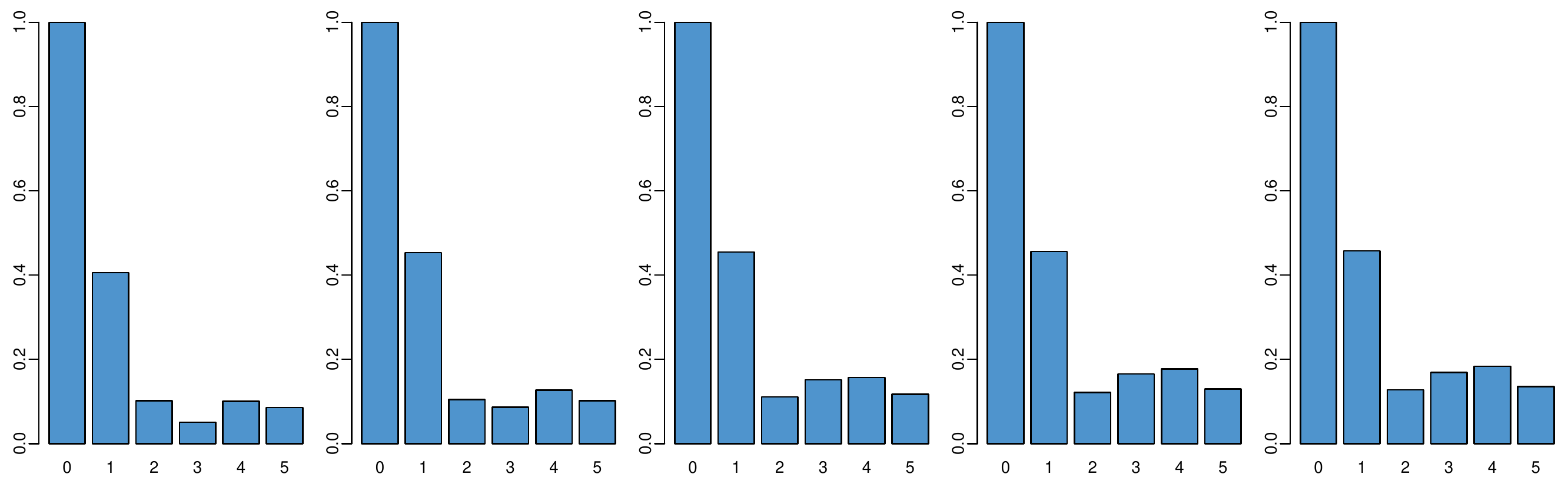}\\
		\caption{Spectral norm of estimated cross-correlation matrices for lags $h=1,\dots,5$ of the vector model based on principal subspaces of dimension $d=1$ to $d=5$ (from left to right).}
		\label{acf}
	\end{center}
\end{figure} 
Consequently, the methodology introduced in Section~\ref{sec:methodology} and \ref{sec:selection} was applied to the data. First, the covariance operator $C_{X;0}$ and its first $15$ eigenelements $(\lambda_1,\nu_1), \dots,(\lambda_{15},\nu_{15})$ were estimated to construct the vector process $(\hat{\mathbf{X}}_j\colon j=1,\ldots,n)$, where $\hat{\mathbf{X}}_j=( \langle X_j , \hat\nu_1\rangle,\dots,\langle X_j , \hat\nu_{15}\rangle)^\top$. Then, the methods described in Sections~\ref{sec:selection} were applied to choose the appropriate dimension $d$ and model order $q$. 

The first four sample eigenfunctions explained 81\% of the variability, hence the TVE criterion with $P=0.8$ gave $d^*=4$ to initialize Algorithm~\ref{IND}. The hypothesis of independence of the left-out score vector process $(\hat{\mathbf{X}}_j[4\!\!:\!\!15]\colon j=1,\ldots,n)$ was rejected with $p$-value $0.03$. Here $\mathbf{X}_j[i\!\!:\!\!i^\prime]$ is used as notation for the vector comprised of coordinates $i,\ldots,i^\prime$, with $i\leq i^\prime$, of the original 15-dimensional vector $\hat{\mathbf{X}}_j$. In the next step of Algorithm \ref{IND}, $d^*$ is increased to $5$. A second independence test was run on $(\hat{\mathbf{X}}_j[5\!:\!15]\colon j=1,\ldots,n)$ and did not result in a rejection; the corresponding $p$-value was $0.25$. 

This analysis led to using $d=5$ as dimension of the principal subspace to conduct model selection with the methods of Section~\ref{sec:AICC}. Since TVE indicated $d=4$, the selection procedures were applied also with this choice. In both cases, the AICC criterion in \eqref{AICC} and LB criterion in \eqref{ljungbox} opted for $q=1$, in accordance with the spectral norms observed in Figure~\ref{acf}. Simultaneously choosing $d$ and $q$ with the fFPE criterion of Section~\ref{sec:fFPE} yields $d=3$ and $q=1$. 

After the model selection step, the operator of the chosen FMA$(1)$ process was estimated using Algorithm \ref{fia}. Similarly the methods introduced in Section~\ref{sec:fma1} were applied. Figure \ref{realtheta} displays the kernels of the estimated integral operator for all methods, selecting for $d=3$ and $d=5$. The plots indicate that, on this particular data set, all three methods produce estimated operators that lead to kernels of roughly similar shape. 
The similarity is also reflected in the covariance of the estimated innovations.  For $d=3$, the trace of the covariance matrix is $43.14$, $45.4$ and $44.41$ for the Innovations Algorithm, iterative method and projective method, respectively. For $d=4$, the trace of the covariance of the estimated innovations is $48.19$, $46.00$ and $45.74$ for the different methods in the same order. 
%This result furthermore confirms the observation that the Innovations Algorithm approach is sensible to high $d$.

\begin{figure}[ht]
	\begin{center}
		% Requires \usepackage{graphicx}[height=7in,width=6.25in]
		\includegraphics[width=\linewidth]{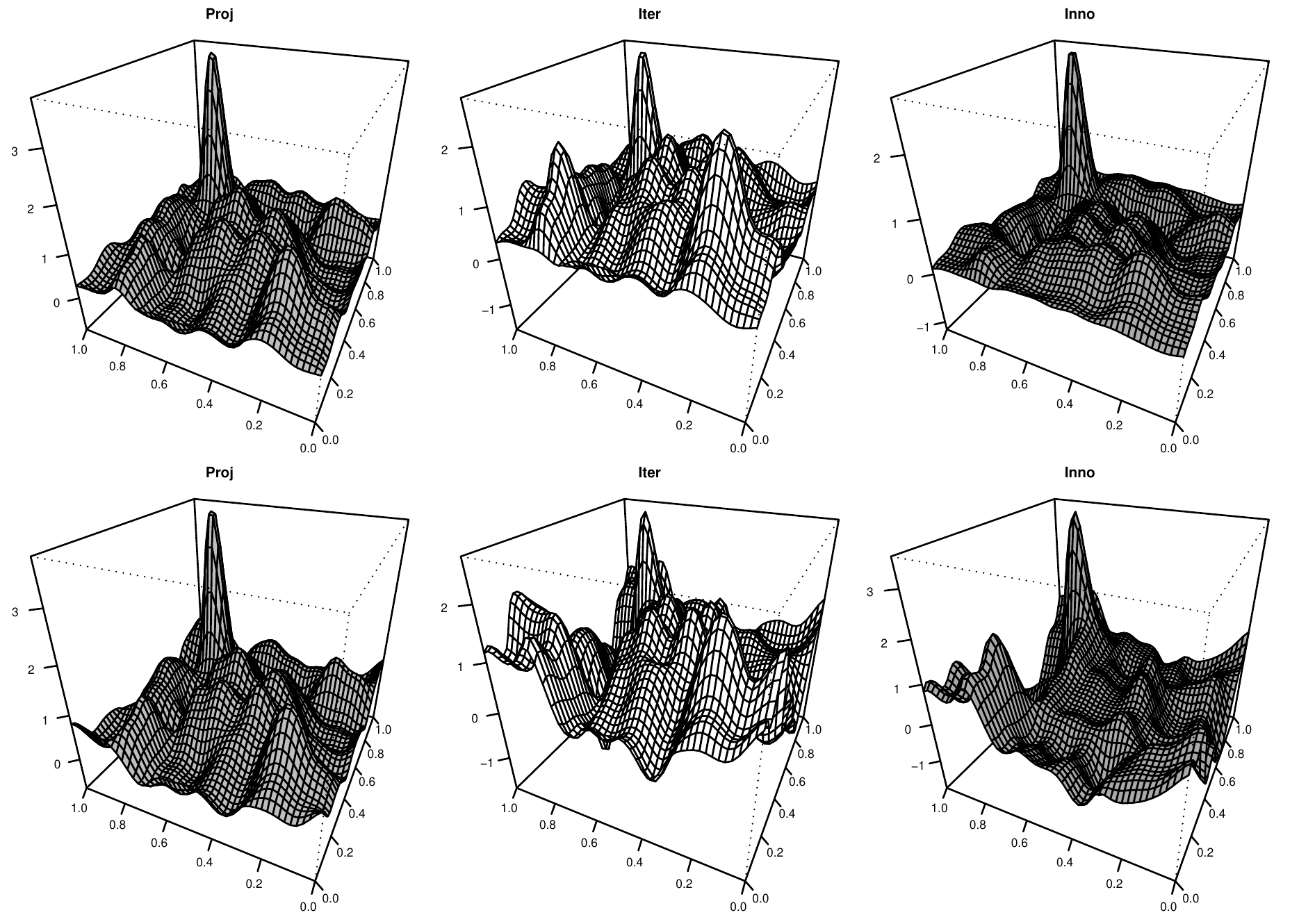}\\
		\caption{Estimated FMA$(1)$ kernel with the three methods for $d=3$ (first row) and $d=4$ (second row)}
		\label{realtheta}
	\end{center}
\end{figure}

%%%%%%%%%%%%%%%%%%%%%%%%
\section{Conclusions}
\label{sec:conclusion}
%%%%%%%%%%%%%%%%%%%%%%%%

This paper is the first to introduce a complete methodology to estimate any stationary, causal and invertible functional time series. This is achieved by approximating the functional linear filters in the causal %invertible
representation with functional moving average processes obtained from an application of the functional Innovations Algorithm. The consistency of the estimators is verified as the main theoretical contribution. The proof relies on the fact that $d$-dimensional projections of FMA($q$) processes are isomorph to $d$ dimensional VMA($q^*$) models, with $q^*\leq q$. Introducing appropriate sequences of increasing subspaces of $H$, consistency can be established in the two cases of known and unknown principal projection subspaces. This line of reasoning follows multivariate techniques given in Lewis and Reinsel \cite{lewis} and Mitchell and Brockwell \cite{mitchell}.

The theoretical underpinnings are accompanied by model selection procedures facilitating the practical implementation of the proposed method. An independence test is introduced to select the dimension of the principal projection subspace, which can be used as a starting point for the suggested order selection procedures based on AICC and Ljung--Box criteria. Additionally, an fFPE criterion is established that jointly selects dimension $d$ and order $q$. Illustrative results from a simulation study and the analysis of traffic velocity data show that the practical performance of the proposed method is satisfactory and at least competitive with other methods available in the literature for the case of FMA(1) processes. 

Future research could focus on an extension of the methodology to FARMA processes in order to increase parsimony in the estimation. It should be noted, however, that this not a straightforward task as identifying the dynamics of the projection of an FARMA$(p,q)$ to a finite-dimensional space is a non-resolved problem. In addition, the proposed methodology could be applied to offer an alternative route to estimate the spectral density operator, a principal object in the study of functional time series in the frequency domain; see Aue and van Delft \cite{avd}, H\"ormann et al.\ \cite{hoermann} and Panaretos and Tavakoli \cite{panaretos}.

%%%%%%%%%%%%%%%%%%%%%%%%
\section{Proofs}
\label{sec:proof}
%%%%%%%%%%%%%%%%%%%%%%%%

The notion of $L^p$-$m$-approximability is utilized for the proofs. A version of this notion was used for multivariate time series in Aue et al.\ \cite{AHHR} and then translated to the functional domain by H\"ormann and Kokoszka~\cite{weaklydep}. The definition is as follows.

\begin{definition}
\label{def:lpm}
{\rm 
Let $p\geq 1$. A sequence $(X_j\colon j\in\mathbb{Z})$ with values in $L^p_H$ is called {\it $L^p$-$m$-approximable}\/ if
\[
X_j=f(\varepsilon_j,\varepsilon_{j-1},\ldots),
\qquad j\in\mathbb{Z},
\]
can be represented as a functional Bernoulli shift with a sequence of independent, identically distributed random elements $(\varepsilon_j\colon j\in\mathbb{Z})$ taking values in the measurable space $S$, potentially different from $H$, and a measurable function $f\colon S^\infty\to H$ such that 
\[
\sum_{m=0}^\infty\big(\E[\|X_j-X_{j}^{(m)}\|^p]\big)^{1/p}<\infty,
\] 
where $X_{j}^{(m)}=f(\varepsilon_j,\ldots,\varepsilon_{j-m+1},\varepsilon_{j-m}^{(j)},\varepsilon_{j-m-1}^{(j)},\ldots)$ with $(\varepsilon_{j}^{(i)}\colon j\in\Z)$, $i\in\N_0$, being independent copies of $(\varepsilon_j\colon j\in\mathbb{Z})$.
}
\end{definition}

Conditions can be established for most of the common linear and nonlinear functional time series models to be $L^p$-$m$-approximable. In particular, the functional linear processes $(X_j\colon j\in\mathbb{Z})$ defined in \eqref{eq:flp} are naturally included if the summability condition $\sum_{m=1}^\infty\sum_{\ell=m}^\infty\|\psi_\ell\|_{\call}<\infty$ is met (see Proposition 2.1 in H\"ormann and Kokoszka \cite{weaklydep}).

\begin{proof}[\bf Proof of Theorem \ref{l4mapp}]
Using that $(X_j\colon j\in\Z)$ is $L^4$-$m$-approximable, write %$X_j=f(\varepsilon_n,\varepsilon_{n-1},\dots)$ and note that $\sum_{m=0}^{\infty} (\E[\Vert X_m - X_m^{(m)}\Vert ^4])^{1/4}<\infty$. Now,
\begin{align*}
X_j(k)  
&=(f(\varepsilon_j,\varepsilon_{j-1},\dots),\ldots,f(\varepsilon_{j-k+1},\varepsilon_{j-k},\dots))^\top \\
& = g(\varepsilon_j,\varepsilon_{j-1},\dots),
\end{align*}
where $g\colon H^\infty\to H^k$ is defined accordingly. For $k,m\in\N$ and $j\in\Z$, define 
\begin{align*}
X_j^{(m)}(k)
&= \big( f(\varepsilon_{j},\dots,\varepsilon_{j-m+1},\varepsilon_{j-m}^{(j)},\varepsilon_{j-m-1}^{(j)},\ldots),\ldots,
%f(\varepsilon_{n-1},\varepsilon_{n-2},\dots,\varepsilon_{n-r+1},\varepsilon_{n-r}^{(n)},\varepsilon_{n-r-1}^{(n)}),\dots,\\
f(\varepsilon_{j-k+1},\dots,\varepsilon_{j-m+1},\varepsilon_{j-m}^{(j)},\varepsilon_{j-m-1}^{(j)},\dots)\big)^\top \\
&=g(\varepsilon_j,\varepsilon_{j-1},\dots,\varepsilon_{j-m+1},\varepsilon_{j-m}^{(j)},\varepsilon_{j-m-1}^{(j)},\ldots).
\end{align*}
%In order to verify $L^4-m$ approximability of $(X_n(k))_{n\in\Z}$, we have to check that $\sum_{r=k}^{\infty} \big(E\Vert X_r(k) - X_r(k)^{(r)}\Vert ^4\big)^{1/4}<\infty$. 
Now, % $E\Vert X_r(k) - X_r(k)^{(r)}\Vert=\sum_{i=0}^{k-1} E\Vert X_{r-i} - X_{r-i}^{(r-i)}\Vert$
by definition of the norm in $H^k$,
\begin{align}
\sum_{m=k}^{\infty} \big(\E\big[\Vert X_m(k) - X_m^{(m)}(k)\Vert ^4\big]\big)^{1/4} 
&= \sum_{m=k}^{\infty}\bigg(\sum_{i=0}^{k-1} \E\big[\Vert X_{m-i}-X_{m-i}^{(m-i)}\Vert^4\big]\bigg)^{1/4} \notag \\
%\end{align*}
%Furthermore by Assumption~\ref{assumptions}, $E\Vert X_n - X_n^{(l-i)}\Vert ^2 \leq E\Vert X_n - X_n^{(l)}\Vert^2$ for  $i\geq 0$, hence
%\begin{align}
%	\sum_{r=k}^{\infty}	 \big( \sum_{i=0}^{k-1} E\Vert X_{r-i}-X_{r-i}^{(r-i)}\Vert^4\big)^{1/4}
&\leq\sum_{m=k}^{\infty}\bigg(\sum_{i=0}^{k-1}\E\big[\Vert X_{m-i}-X_{m-i}^{(m-k)}\Vert^4\big]\bigg)^{1/4}\notag \\
&=\sum_{m=k}^{\infty}\big(k \E\big[\Vert X_{m-k}-X_{m-k}^{(m-k)}\Vert^4\big]\big)^{1/4} \notag\\
		&= k ^{1/4}\sum_{m=0}^{\infty} \big(\E\big[\Vert X_{m}-X_{m}^{(m)}\Vert^4\big]\big)^{1/4}, \label{xrk}
\end{align}
where the first inequality is implied by Assumption \ref{assumptions}, since $\E[\Vert X_j - X_j^{(m-i)}\Vert ^2] \leq \E[\Vert X_j - X_j^{(m)}\Vert^2]$ for all $i\geq 0$, and the last inequality, since $\E[\Vert X_1 - X_1^{(m-k)}\Vert ^2] = \E[ \Vert X_{j} - X_j^{(m-k)}\Vert ^2 ]$ by stationarity. But the right-hand side of \eqref{xrk} is finite because $(X_j\colon j\in\Z)$ is $L^4$-$m$-approximable by assumption. This shows that $(X_j(k)\colon j\in\Z)$ is also $L^4$-$m$ approximable.
		
To prove the consistency of the estimator $\wh C_{X(k)}$, note that the foregoing implies, by Theorem~3.1 in H\"ormann and Kokoszka~\cite{weaklydep}, that the bound
\begin{align*}
n\,\E\big[\Vert \wh C_{X(k)} - C_{X(k)}\Vert_\caln^2\big] \leq U_{X(k)},
\end{align*}
holds, where $U_{X(k)}= \E[\Vert X_1(k) \Vert^4]+ 4 \sqrt{2}(\E[\Vert X_1(k) \Vert^4])^{3/4} \sum_{m=0}^{\infty} (\E[\Vert X_m(k)- X_m^{(m)}(k) \Vert^4])^{1/4}$ is a constant that does not depend on $n$. Since $\E[\Vert X_1(k) \Vert ^4] =k\E[\Vert X_1\Vert ^4]$, \eqref{xrk} yields that $U_{X(k)}=kU_X$,
%\begin{align*}
%U_{X(k)} = k \E \Vert X \Vert ^4 + 4\sqrt{2} k ^{3/4} \, (\E\Vert X \Vert ^4 )^{3/4}  \, k^{1/4} \, \sum_{r=1}^{\infty} ( \E \Vert X_r - X_r^{(r)}\Vert ^4 ) ^{1/4} = k\, U_X,
%\end{align*}
which is the assertion.
\end{proof}

%%%%%%%%%%%%%%%%%%%%%%%%

%\subsection{Proof of Proposition~\ref{autoregressive}}
%For what follows, we require the following Proposition. 
%
%\begin{proposition}[\cite{kk}, Proposition~6.4]\label{prophelp}
%	Under the Assumptions~\ref{assumptions} (i), for every $i, j \in \mathbb{N},j\leq q$ and $x\in H$, as $k\rightarrow \infty$,
%	\begin{compactenum}
%		\item[(i)] $\big\Vert(\beta_{k,i} - \pi_i)( x ) \big\Vert \rightarrow 0$,
%		\item[(ii)] $\big\Vert(\beta_{k,i} - \beta_{k-j,i})(x) \big\Vert \rightarrow 0$,
%		\item[(iii)] $\big\Vert(\theta_{k,i} - \theta_{k-j,i} )(x)\big\Vert \rightarrow 0$,
%		\item[(iv)] $\Vert (\theta_{k,i} - \gamma_i)(x) \Vert\rightarrow 0.$
%	\end{compactenum}
%\end{proposition}

\begin{cor}\label{corhelp}
	The operators $\wh\beta_{k,i}$ from \eqref{yulewalkerhat} and $\wh \theta_{k,i}$ from \eqref{innov} related through
	\begin{align}
	\wh\theta_{k,i}=\sum_{j=1}^i \wh\beta_{k,j} \wh\theta_{k-j,i-j}, \qquad   i=1,\dots,k,\; k\in\N. 
	\label{linkhat}
	\end{align}
\end{cor}

\begin{proof}[\bf Proof]
The proof is based on the finite-sample versions of the regression formulation of \eqref{blp} and the innovations formulation given in \eqref{innov}. Details are omitted to conserve space. %The following two equations are the finite sample versions of their population counterparts \eqref{blp} and \eqref{innov}: %\JK{[should be stated in the beginning, and only a reference should be given here]}
%	\begin{align}
%	\wh X_{n+1}^k &= \sum_{i=1}^{k} \wh \beta_{k,i} X_{d_{k+1-i},n+1-i}, \quad \text{and} \label{pred1}\\
%	\wh X_{n+1}^k &= \sum_{i=1}^{k} \wh \theta_{k,i} ( X_{d_{k+1-i},n+1-i}- \wh X_{n+1-i}^{k-i}), \quad n,k\in\N. \label{pred2}
%	\end{align}
%	Define $\wh\theta_{k,0}=I_H$. Then adding $\wh \theta_{k,0}(X_{d_{k+1},n+1}- \wh X_{n+1}^{k})$ to both sides of \eqref{pred2} yields
%	\begin{align*}
%	X_{d_{k+1},n+1} &= \sum_{i=0}^{k} \wh \theta_{k,i} ( X_{d_{k+1-i},n+1-i}- \wh X_{n+1-i}^{k-i}), \quad n,k\in\N.
%	\end{align*}
%	Plugging that expression for $	X_{d_{k+1},n+1}$ in \eqref{pred1} yields
%	\begin{align}
%	\wh X_{n+1}^k &= \sum_{i=1}^{k} \wh \beta_{k,i} \sum_{j=0}^{k-i} \wh \theta_{k-i,j} ( X_{d_{k-i+1-j},n-i+1-j}- \wh X_{n-i+1-j}^{k-i-j}), \quad n,k\in\N. \label{pred3}
%	\end{align}
%	Now equating the coefficients in \eqref{pred2} and \eqref{pred3} yields the desired result.
\end{proof}

%We are now ready to proof Theorem~\ref{autoregressive}.
\begin{proof}[\bf Proof of Theorem \ref{autoregressive}]
{\em (i)}
%	We follow the lines of the proof of \cite{lewis}. 
It is first shown that, for all $x\in H^k$, 
\begin{align*}
\Vert 	(\wh B (k) - \Pi(k) )(x) \Vert \overset{p}{\rightarrow} 0 \qquad  (n\rightarrow \infty),
\end{align*}
where $\Pi(k)=(\pi_1,\ldots,\pi_k)^\top$ is the vector of the first $k$ operators in the invertibility representation of the functional time series $(X_j\colon j\in\Z)$. 
Define the process $(e_{j,k}\colon j\in\Z)$ by letting
\begin{align} \label{elk}
e_{j,k}=X_j-\sum_{\ell=1}^k\pi_\ell X_{j-\ell}
\end{align} 
and let $I_{H^k}$ be the identity operator on $H^k$. Note that
\begin{align*}
\wh B(k) - \Pi(k) 
&= \wh \Gamma_{1,k,d} \wh\Gamma_{k,d}^{-1} - \Pi(k) \wh\Gamma_{k,d}\wh\Gamma_{k,d}^{-1} +   \Pi(k)(I_{H^k}-P_{(k)})\\
&= \big(\Gamma_{1,k,d}-\Pi(k) \wh\Gamma_{k,d}\big)\wh\Gamma_{k,d}^{-1}+\Pi(k)(I_{H^k}-P_{(k)}).
\end{align*}
Plugging in the estimators defined in \eqref{hatgammakd} and subsequently using \eqref{elk}, it follows that
\begin{align*}
\wh B(k) - \Pi(k) 
&= \bigg(\frac{1}{n-k}\sum_{j=k}^{n-1} \big((P_{(k)}X_{j,k}\otimes X_{j+1}) -  (P_{(k)} X_{j,k}\otimes\Pi(k) X_{j,k})\big)\bigg)\wh\Gamma_{k,d}^{-1}+\Pi(k)(I_{H^k}-P_{(k)})\\
&=\bigg(\frac{1}{n-k}\sum_{j=k}^{n-1} \big(P_{(k)}X_{j,k}\otimes (X_{j+1} - \Pi(k) X_{j,k})\big)\bigg)\wh\Gamma_{k,d}^{-1}+\Pi(k)(I_{H^k}-P_{(k)})\\
&=\bigg(\frac{1}{n-k}\sum_{j=k}^{n-1} \big(P_{(k)}X_{j,k}\otimes e_{j+1,k}\big)\bigg)\wh\Gamma_{k,d}^{-1}+\Pi(k)(I_{H^k}-P_{(k)}).
\end{align*}
Two applications of the triangle inequality imply that, for all $x\in H^k$,
\begin{align}
\Vert 	(\wh B(k) - \Pi(k)) (x)  \Vert 
&\leq\bigg\Vert\bigg(\frac{1}{n-k}\sum_{j=k}^{n-1} \big(P_{(k)}X_{j}(k)\otimes e_{j+1,k}\big)\bigg)\wh\Gamma_{k,d}^{-1} (x)\bigg\Vert
+ \Vert\Pi(k)(I_{H^k}-P_{(k)})(x)\Vert \notag \\
%\end{align}
%The first summand can be treated uniformly for all $x\in H^k$ with $\Vert x\Vert \leq 1$, but the second has to be considered pointwise. Using again the triangle inequality, we get for all $x\in H^k$
%	\begin{align}
%	\Vert 	(\wh B(k) - \Pi(k)) (x)  \Vert
&\leq \bigg\Vert\bigg(\frac{1}{n-k}\sum_{j=k}^{n-1} \big(P_{(k)}X_{j}(k)\otimes (e_{j+1,k}-\varepsilon_{j+1})\big)\bigg)\wh\Gamma_{k,d}^{-1}\bigg\Vert_{\call}\notag\\
&\qquad + \bigg\Vert\bigg(\frac{1}{n-k}\sum_{j=k}^{n-1} \big(P_{(k)}X_{j}(k)\otimes \varepsilon_{j+1}\big)\bigg)\wh\Gamma_{k,d}^{-1}\bigg\Vert_{\call}+\Vert\Pi(k)(I_{H^k}-P_{(k)})(x)\Vert\notag\\
&\leq \big(\Vert U_{1n}\Vert_{\call} + \Vert U_{2n}\Vert_{\call}\big)\Vert \wh\Gamma_{k,d} ^{-1}\Vert_{\call}+\Vert\Pi(k)(I_{H^k}-P_{(k)})(x)\Vert, \label{help4}
\end{align}
where $U_{1n}$ and $U_{2n}$ have the obvious definitions. Arguments similar to those used in Proposition~6.4 of Klepsch and Kl\"uppelberg~\cite{kk} yield that the second term on the right-hand side of \eqref{help4} can be made arbitrarily small by increasing $k$. To be more precise, for $\delta>0$, there is $k_\delta\in\N$ such that
\begin{align}\label{res6}
\Vert\Pi(k)(I_{H^k}-P_{(k)})(x)\Vert<\delta
\end{align}
for all $k\ge k_\delta$ and all $x\in H^k$.
	
To estimate the first term on the right-hand side of \eqref{help4}, focus first on $ \Vert \wh\Gamma_{k,d} ^{-1}\Vert_{\call} $. Using the triangular inequality, 
%	\begin{align}
$\Vert \wh\Gamma_{k,d} ^{-1}\Vert_{\call} \leq \Vert \wh \Gamma_{k,d}^{-1} - \Gamma_{k,d}^{-1} \Vert_{\call}+\Vert \Gamma_{k,d}^{-1}\Vert_{\call}.$ %\label{res0}
%	\end{align}
Theorem 1.2 in Mitchell \cite{mitchell2} and Lemma~6.1 in Klepsch and Kl\"uppelberg \cite{kk} give the bound
\begin{align}
\Vert \Gamma_{k,d}^{-1}\Vert_{\call} \leq \alpha_{d_k}^{-1},
\label{res1}
\end{align}
where $\alpha_{d_k}$ is the infimum of the eigenvalues of all spectral density operators of $(X_{d_k,j}\colon j\in\Z)$. Furthermore, using the triangle inequality and then again Lemma~6.1 of Klepsch and Kl\"uppelberg \cite{kk}, 
\begin{align}
\Vert \wh \Gamma_{k,d}^{-1} - \Gamma_{k,d}^{-1} \Vert_{\call} 
&= \Vert \wh \Gamma_{k,d}^{-1} (\wh \Gamma_{d,k} - \Gamma_{d,k}) \Gamma_{k,d}^{-1} \Vert_{\call} \notag \\
&\leq \big( \Vert \wh \Gamma_{k,d}^{-1} - \Gamma_{k,d}^{-1} \Vert_{\call} + \Vert  \Gamma_{k,d}^{-1} \Vert_{\call} \big)  \Vert\wh \Gamma_{d,k} - \Gamma_{d,k}\Vert_{\call} \alpha_{d_k}^{-1}. 
\label{help1}
\end{align}
Hence, following arguments in the proof of Theorem~1 in Lewis and Reinsel \cite{lewis},
\begin{align*}
0 \leq \frac{\Vert \wh \Gamma_{k,d}^{-1} - \Gamma_{k,d}^{-1} \Vert_{\call} }{\alpha_{d_k}^{-1} (\Vert \wh \Gamma_{k,d}^{-1} - \Gamma_{k,d}^{-1} \Vert_{\call} +\alpha_{d_k}^{-1})} \leq \Vert\wh \Gamma_{d,k} - \Gamma_{d,k}\Vert_{\call} ,
\end{align*}
by \eqref{help1}. This yields
\begin{align}\label{help2}
\Vert\wh \Gamma_{d,k}^{-1} - \Gamma_{d,k}^{-1}\Vert_{\call} \leq \frac{\Vert\wh \Gamma_{d,k} - \Gamma_{d,k} \Vert_{\call} \alpha_{d_k}^{-2}}{1-\Vert\wh \Gamma_{d,k} - \Gamma_{d,k}\Vert_{\call}\alpha_{d_k}^{-1}}. 
\end{align}
Note that, since $P_{(k)}P_k=P_{(k)}$,
%	\begin{align*}
$\Vert \Gamma_{k,d}\Vert_{\call} = \Vert  P_{(k)}P_k \Gamma_{k} P_k P_{(k)} \Vert_{\call} \leq \Vert P_{k} \Gamma_{k} P_{k}\Vert_{\call}.$
%\end{align*}
Also, by Theorem~\ref{l4mapp}, for some positive finite constant $M_1$, 
%	\begin{align*}
 $\E[ \Vert P_k \wh\Gamma_k P_k - P_k \Gamma_k P_k \Vert^2] \leq M_1 k/{(n-k)}.$
%	\end{align*}
Therefore,
\begin{align}
\Vert  \wh\Gamma_{d,k}- \Gamma_{d,k} \Vert = O_p\bigg(\sqrt{\frac{k}{n-k}}\bigg).
\end{align}
Hence, the second part of Assumption~\ref{assumptions} %as $N\rightarrow\infty$, the denominator of the right-hand-side of \eqref{help2} satisfies $\Vert\wh \Gamma_{d,k} - \Gamma_{d,k} \Vert_{\call} \, \alpha_{d_k}^{-2} \overset{p}{\rightarrow} 0$.
and \eqref{help2} lead first to
%	 \begin{align} \label{res2}
$\Vert\wh \Gamma_{d,k}^{-1} - \Gamma_{d,k}^{-1}\Vert_{\call} \overset{p}{\rightarrow} 0$ %\quad N\rightarrow \infty$
%	\end{align}
and, consequently, combining the above arguments, 	%Combining \eqref{res0}, \eqref{res1} and \eqref{res2}, we get
\begin{align}
\Vert \wh\Gamma_{k,d} ^{-1}\Vert_{\call} = O_p(\alpha_{d_k}^{-1}). \label{res5}
\end{align}
	
Next consider $U_{1n}$ in \eqref{help4}. With the triangular and Cauchy--Schwarz inequalities, calculate
\begin{align*}
\E[\Vert U_{1n}\Vert] 
&=\E\bigg[\bigg\Vert \frac{1}{n-k} \sum_{j=k}^{n-1} P_{(k)}X_{j}(k) \otimes (e_{j+1,k}-\varepsilon_{j+1}) \bigg\Vert_{\call}\bigg] \\ 
&\leq \frac{1}{n-k} \sum_{j=k}^{N-1}\E\bigg[\bigg\Vert  P_{(k)}X_{j}(k) \otimes (e_{j+1,k}-\varepsilon_{j+1})\bigg\Vert_{\call}\bigg] \\ 
&\leq \frac{1}{n-k} \sum_{j=k}^{N-1} \big(\E[\Vert  P_{(k)}X_{j}(k)\Vert^2] \big)^{1/2} \big(\E[\Vert e_{j+1,k}-\varepsilon_{j+1}\Vert^2] \big)^{1/2}.
\end{align*}
The stationarity of $(X_j\colon j\in\Z)$ and the fact that $X_j\in L^2_H$ imply that, for a positive finite constant $M_2$, 
\begin{align}
\E[\Vert U_{1n}\Vert_\call] 
& \leq\big(\E[\Vert  P_{(k)} X_{j}(k) \Vert^2] \big)^{1/2} \big(\E[\Vert e_{j+1,k} -\varepsilon_{j+1} \Vert^2] \big)^{1/2} \notag \\
&\leq \sqrt{k} \big( \E[\Vert P_{\mathcal{V}_{d_k}} X_0 \Vert ^2] \big)^{1/2}  
\bigg(\E\bigg[ \bigg\Vert \sum_{\ell>k} \pi_\ell X_{1-\ell}
+ \sum_{\ell=1}^{k} \pi_{\ell}(I_H-P_{\mathcal{V}_{d_{k+1-\ell}}})X_{1-\ell}\bigg\Vert^2\bigg]\bigg)^{1/2} \notag\\
&\leq \sqrt{k} \bigg(2 \E\bigg[ \Vert \sum_{\ell>k} \pi_\ell X_{j+1-\ell} \Vert^2\bigg] 
+ 2\E\bigg[\bigg\Vert \sum_{\ell=1}^{k} \pi_{\ell} (I_H-P_{\mathcal{V}_{d_{k+1-\ell}}}) X_{1-\ell} \bigg\Vert^2\bigg]\bigg)^{1/2} \notag \\
&= M_2\sqrt{k(J_1 + J_2)} \notag \\
& \leq M_2\sqrt{k}(\sqrt{J_1} + \sqrt{J_2}), \label{help5}
\end{align}
where $J_1$ and $J_2$ have the obvious definition. Since for $X\in L^2_H$,  $\E[ \Vert X\Vert^2] = \Vert C_X\Vert_\caln$, the term $J_1$ can be bounded as follows. Observe that
\begin{align*}
J_1 
&=\bigg\Vert \E\bigg[ \sum_{\ell>k} \pi_\ell X_{1-\ell} \otimes  \sum_{\ell'>k} \pi_{\ell'} X_{1-\ell'} \bigg]\bigg\Vert_{\caln} \\
&=\bigg\Vert\sum_{\ell,\ell'>k} \pi_\ell C_{X;\ell-\ell'} \pi_{\ell'}^* \bigg\Vert_{\caln} \\
&\leq \sum_{\ell,\ell'>k} \Vert \pi_\ell\Vert_{\call} \Vert \pi_{\ell^\prime}\Vert_{\call} \Vert C_{X;\ell-\ell'} \Vert_{\caln}.
\end{align*}
Now $C_{X;\ell-\ell'}\in\caln$ for all $\ell,\ell'\in\Z$, hence $\Vert C_{X;\ell-\ell'} \Vert_{\caln}\leq M_3$ and $J_1\leq M_3 (\sum_{\ell>k} \Vert \pi_\ell\Vert_{\call})^2$. Concerning $J_2$, note first that, since $\E[ \Vert X\Vert^2] = \Vert C_X\Vert_\caln$,
\begin{align}
J_2 &= \bigg\Vert  \E \bigg [\sum_{\ell=1}^k \pi_\ell( I_H-P_{\mathcal{V}_{d_{k+1-\ell}}})X_{1-\ell} \otimes \sum_{\ell'=1}^n \pi_{\ell'} ( I_H-P_{\mathcal{V}_{d_{k+1-\ell'}}})X_{1-\ell'} \bigg] \bigg\Vert_{\caln}. \notag
\end{align}
Using the triangle inequality together with properties of the nuclear operator norm and the definition of $C_{X;h}$ in display \eqref{cxh} leads to 
\begin{align}
% E\big \Vert \sum_{i=1}^n \pi_i ( I_H-P_{A_{d_{n+1-i}}})X_{n+1-i}\big \Vert^2
J_2
&\leq \sum_{\ell,\ell'=1}^k  \Vert \pi_\ell\Vert_{\call} \Vert \pi_\ell'\Vert_{\call} \big\Vert \E \big[ ( I_H-P_{\mathcal{V}_{d_{k+1-\ell}}})X_{1-\ell} \otimes  ( I_H-P_{\mathcal{V}_{d_{k+1-\ell'}}})X_{1-\ell'}\big] \big\Vert_{\caln} \notag\\
&= \sum_{\ell,\ell'=1}^k  \Vert \pi_\ell\Vert_{\call} \Vert \pi_\ell'\Vert_{\call} \big\Vert ( I_H-P_{\mathcal{V}_{d_{k+1-\ell}}})C_{X;\ell-\ell'} ( I_H-P_{\mathcal{V}_{d_{k+1-\ell'}}}) \big\Vert_{\caln} \notag \\
&= \sum_{\ell,\ell'=1}^k  \Vert \pi_\ell\Vert_{\call} \Vert \pi_\ell'\Vert_{\call}  K(\ell,\ell').\label{proof1}
\end{align}
By the definition of $\mathcal{V}_d$ in \eqref{xdi} and since $( I_H-P_{\mathcal{V}_{d_{i}}})= \sum_{r>d_{i}} \nu_r\otimes\nu_r$, it follows that
\begin{align}
%	\big\Vert ( I_H-P_{A_{d_{n+1-i}}})C_{X;i-j} ( I_H-P_{A_{d_{n+1-j}}}) \big\Vert_{\caln} 
K(\ell,\ell') 
%&= \bigg\Vert \bigg(\sum_{s>d_{k+1-\ell'}} \nu_{s}\otimes \nu_{s}\bigg) C_{X;\ell-\ell'}\bigg(\sum_{r>d_{k+1-\ell}} \nu_r\otimes \nu_r\bigg)\bigg\Vert_{\caln}
%	&=\big\Vert ( I_H-P_{A_{d_{n+1-i}}}) \sum_{l>d_{n+1-j}} \nu_l\otimes C_{X;i-j}(\nu_l)\big\Vert_{\caln}\\
&=\bigg\Vert\sum_{s>d_{k+1-\ell'}} \sum_{r>d_{k+1-\ell}} \langle C_{X;\ell-\ell'} (\nu_r), \nu_{s} \rangle \nu_r\otimes \nu_{s}\bigg\Vert_{\caln} \notag \\
%\end{align*}
%	\begin{align}
	%\big\Vert ( I_H-P_{A_{d_{n+1-i}}})C_{X;i-j} ( I_H-P_{A_{d_{n+1-j}}}) \big\Vert_{\caln} 
%	K(l,l') 
&\leq \bigg\Vert\sum_{s>d_{k+1-\ell'}} \sum_{r>d_{k+1-\ell}} \sqrt{\lambda_r\lambda_{s}} \nu_r\otimes \nu_{s}\bigg\Vert_{\caln} \notag\\
&= \sum_{i=1}^{\infty} \bigg\langle \sum_{s>d_{k+1-\ell'}} \sum_{r>d_{k+1-\ell}} \sqrt{\lambda_r\lambda_{s}} \nu_r\otimes \nu_{s} (\nu_i), \nu_i \bigg\rangle \notag \\
%&= \sum_{k>d_{n+1-j}} \big\langle \sum_{l'>d_{n+1-i}} \lambda_k^{1/2}\lambda_{l'}^{1/2} \nu_{l'} , \nu_k \big\rangle \notag\\ 
%&= \sum_{i>\max(d_{k+1-l},d_{k+1-l'})} \lambda_i 
& \leq \sum_{i>d_{k+1-\ell}} \lambda_i, \label{proof2}
	\end{align}
where Lemma~6.2 in Klepsch and Kl\"uppelberg \cite{kk} was applied to give $\langle C_{X;\ell-\ell^\prime} \nu_r,\nu_s \rangle  \leq \sqrt{\lambda_r \lambda_s}$. Plugging \eqref{proof2} into \eqref{proof1}, and recalling that $\sum_{\ell=1}^{\infty} \Vert\pi_\ell\Vert_{\call}=M_4<\infty$, gives that
\begin{align}
%E\big \Vert \sum_{i=1}^n \pi_i ( I_H-P_{A_{d_{n+1-i}}})X_{n+1-i}\big \Vert^2 
J_2 &\leq M_4 \sum_{\ell=1}^k \Vert \pi_\ell \Vert_{\call} \sum_{i>d_{k+1-\ell}} \lambda_i. \label{proof6}
\end{align}
Inserting the bounds for $J_1$ and $J_2$ into \eqref{help5}, for some $M<\infty$,
\begin{align}
\E[\Vert U_{1n} \Vert]
& \leq \sqrt{k} M_2 (M_3 \sqrt{J_1}+ \sqrt{J_2}) \notag \\ 
&\leq \sqrt{k} M_2  \bigg(M_3 \sum_{\ell>k} \Vert \pi_\ell\Vert_{\call} + M_4 \sum_{\ell=1}^k \Vert \pi_\ell \Vert_{\call} \sum_{i>d_{k+1-\ell}} \lambda_i\bigg) \notag\\
&\leq \sqrt{k}M\bigg(\sum_{\ell>k} \Vert \pi_\ell\Vert_{\call} + \bigg(\sum_{\ell=1}^k \Vert \pi_\ell \Vert_{\call} \sum_{i>d_{k+1-\ell}} \lambda_i ) \bigg) \label{help7}.
\end{align}
	
Concerning $U_{2n}$ in \eqref{help4}, use the linearity of the scalar product, the independence of the innovations $(\varepsilon_j\colon j\in\mathbb{Z})$ and the stationarity of the functional time series $(X_j\colon j\in\mathbb{Z})$ to calculate 
\begin{align*}
E[\Vert U_{2n} \Vert^2] 
%&= \Big(\frac{1}{N-k}\Big)^2 \sum_{j,j'=k}^{N-1} E \big\langle P_{(k)}X_{j}(k)\otimes \varepsilon_{j+1}, P_{(k)}X_{j'}(k)\otimes \varepsilon_{j'+1} \big\rangle\\
%	&=	\Big(\frac{1}{N-k}\Big)^2 \sum_{j,j'=k}^{N-1} E \Big( \big\langle \langle P_{(k)}X_{j}(k) , \cdot\rangle P_{(k)}X_{j'}(k),\cdot\big\rangle \langle\varepsilon_{j+1},\varepsilon_{j'+1}\rangle\Big)
%	\end{align*}
%	The independence of $(\varepsilon_n)_{n\in\Z}$ then yields
%	\begin{align*}
%	E\Vert U_{2N} \Vert^2 
%	&=\Big(\frac{1}{N-k}\Big)^2 \sum_{j=k}^{N-1} E \big\langle \langle P_{(k)}X_{j}(k) , \cdot\rangle P_{(k)}X_{j'}(k),\cdot\big\rangle E\Vert\varepsilon_{j+1}\Vert^2\\
&\leq \bigg(\frac{1}{n-k}\bigg)^2 \sum_{j=k}^{n-1} \E\big[ \Vert P_{(k)}X_{j}(k)\Vert ^2\big] \E\big[\Vert\varepsilon_{j+1}\Vert^2\big] \\
%\end{align*}
%	Using the stationarity of $(X_n)_{n\in\Z}$ and $(\varepsilon_n)_{n\in\Z}$ and then the definition of $(X_j(k))$ yields
%	\begin{align*}
%	E\Vert U_{2N} \Vert^2 
&\leq \frac{1}{n-k} \E\big[ \Vert P_{(k)}X_{0}(k)\Vert ^2\big] \E\big[\Vert\varepsilon_{0}\Vert^2\big] \\
&\leq \frac{k}{n-k} \E\big[ \Vert X_{0}\Vert ^2\big] \E\big[\Vert\varepsilon_{0}\Vert^2\big].
\end{align*}
Since both $(X_j\colon j\in\Z)$ and $(\varepsilon_j\colon j\in\Z)$ are in $L^2_H$, \eqref{res5} implies that
\begin{align*}
\Vert U_{2n} \Vert_\call \Vert \Vert \wh\Gamma_{k,d} ^{-1}\Vert_{\call} 
= O_p\bigg(\frac{1}{\alpha_{d_k}} \sqrt{\frac{k}{n-k}}\bigg).
	\end{align*}
Furthermore, \eqref{res5} and \eqref{help7} show that
\begin{align*}
\Vert U_{1n} \Vert_\call \Vert \Vert \wh\Gamma_{k,d} ^{-1}\Vert_{\call} 
= O_p\bigg(\frac{\sqrt{k}}{\alpha_{d_k}} \bigg(\sum_{\ell>k} \Vert \pi_\ell\Vert_{\call} 
+ \sum_{\ell=1}^k \Vert \pi_\ell \Vert_{\call} \sum_{i>d_{k+1-\ell}} \lambda_i  \bigg)\bigg).
\end{align*}
Thus Assumption~\ref{assumptions}, \eqref{help4} and \eqref{res6} assert that, for all $x\in H^k$,
%\begin{align*}
$\Vert 	\wh B_k - \Pi(k) (x) \Vert \overset{p}{\rightarrow} 0$, %\quad N\rightarrow \infty.
%$\end{align*}
which proves the first statement of the theorem.

{\em (ii)} First note that, for all $x\in H^k$,
%\begin{align*}
$\Vert (\wh\beta_{k,i}-\beta_{k,i})(x)\Vert \leq \Vert  (\wh\beta_{k,i} - \pi_i)(x)\Vert + \Vert (\pi_i-\beta_{k,i})(x)\Vert \overset{p}{\rightarrow } 0$
%\end{align*}
as $n\to\infty$. Now %by Proposition~\ref{prophelp}(i) 
$\theta_{k,1} = \beta_{k,1}$ and by Corollary~\ref{corhelp} $\wh\theta_{k,1}=\wh\beta_{k,1}$. Since furthermore $\sum_{j=1}^{k} \pi_j \psi_{k-j}=\psi_k$ (see, for instance, the proof of Theorem~5.3 in Klepsch and Kl\"uppelberg \cite{kk}), $\psi_1=\pi_1$. Therefore,
\begin{align*}
\Vert (\wh\theta_{k,1}-\psi_1)(x)\Vert = \Vert(\wh\beta_{k,1}-\pi_1)(x)\Vert \overset{p}{\rightarrow} 0
\end{align*}
as $n\to\infty$. This proves the statement for $i=1$. %by Proposition~\ref{autoregressive} for $i=1$. This proves the statement for $i=1$. 
Proceed by assuming the statement of the theorem is true for $i=1,\dots,N\in\N$, and then use induction on $N$. Indeed, for $i=N+1$,  the triangle inequality yields, for all $x\in H$,
\begin{align*}
\Vert (\wh\theta_{k,N+1}-\psi_{N+1})(x)\Vert
& = \bigg\Vert \bigg(\sum_{j=1}^{N+1} \wh\beta_{k,j} \wh\theta_{k-j,N+1-j}  -  \pi_j \psi_{N+1-j}\bigg)(x)\bigg\Vert \\
&\leq \sum_{j=1}^{N+1}  \Vert  (\wh\beta_{k,j}-\pi_j) \wh\theta_{k-j,N+1-j}(x)\Vert + \Vert \pi_j(\wh\theta_{k-j,N+1-j}- \psi_{N+1-j})(x)\Vert.
\end{align*}
Now, for $n\rightarrow\infty$, the first summand converges in probability to $0$ by part {\em (i)}, while the second summand converges to $0$ in probability by induction. Therefore the statement is proven.
\end{proof}

\begin{proof}[\bf Proof of Theorem~\ref{autoregressive2}]
{\em (i)} The proof is based again on showing that, for all $x\in H^k$, $\Vert (\wh{\wh B}(k) - \Pi (k)) (x) \Vert \overset{p}{\rightarrow} 0$ as $n\rightarrow\infty$, where $\wh {\wh B}(k)= (\wh{\wh \beta}_{k,1}, \dots, \wh {\wh \beta}_{k,k})$. To this end, first note that
\begin{align}
\Vert (\wh{\wh B}(k) - \Pi(k))(x) \Vert 
\leq 	\Vert (\wh{\wh B}(k) - \wh B (k) )(x) \Vert 
+ \Vert (\wh B(k) - \Pi(k)) (x) \Vert. 
\label{ausgang2}
\end{align}
Under Assumptions~\ref{assumptions}, the second term of the right-hand side converges to $0$ in probability for all $x\in H^k$ by part {\em (i)} of Theorem~\ref{autoregressive}. The first term of the right-hand side of \eqref{ausgang2} can be investigated uniformly over $ H^k$. Using the plug-in estimators defined as in \eqref{yulewalkerhathat},  we get for $k\in\N$
\begin{align}
\Vert \wh{\wh B}(k) - \wh B (k)  \Vert_\call 
&= \Vert  \wh{\wh\Gamma}_{1,k,d}\wh{\wh \Gamma}_{k,d}^{-1}-\wh \Gamma_{1,k,d} \wh { \Gamma}_{k,d}^{-1}\Vert_\call\notag\\
&\leq \Vert \big( \wh{\wh\Gamma}_{1,k,d} -\wh \Gamma_{1,k,d} \big) \wh{\wh \Gamma}_{k,d}^{-1}\Vert_\call 
+ \Vert \wh \Gamma_{1,k,d} \big( \wh{ \Gamma}_{k,d}^{-1}-\wh {\wh \Gamma}_{k,d}^{-1}\big)\Vert_\call. \label{ausgang1}
\end{align}
Following the same intuition as in the proof of Theorem~\ref{autoregressive}, start by investigating the term $ \Vert( \wh{ \Gamma}_{k,d}-\wh {\wh \Gamma}_{k,d})\Vert_\call $. Applying triangle inequality, linearity of the inner product and the inequalities $\Vert P_{(k)}X_{j}(k)\Vert \leq \Vert X_j(k)\Vert$ and $\Vert \wh P_{(k)}X_{j}(k)\Vert \leq \Vert X_j(k)\Vert$, it follows that
\begin{align}
\Vert( \wh{ \Gamma}_{k,d}-\wh {\wh \Gamma}_{k,d})\Vert_\call 
&= \bigg\Vert \frac{1}{n-k} 
\sum_{j=k}^{n-1} \big(P_{(k)}X_j(k)\otimes P_{(k)}X_j(k) - \wh P_{(k)}X_j(k)\otimes \wh P_{(k)}X_j(k) \big)\bigg\Vert_\call\notag\\
%&\leq  \frac{1}{n-k} \sum_{j=k}^{n-1} \big\Vert P_{(k)}X_j(k) \big \Vert  \big\Vert P_{(k)}X_j(k)-\wh P_{(k)}X_j(k) \big\Vert \notag\\
%&\qquad+ \frac{1}{n-k} \sum_{j=k}^{n-1} \big\Vert P_{(k)}X_j(k)- \wh P_{(k)}X_j(k)\big\Vert \big\Vert \wh P_{(k)}X_j(k) \big\Vert\notag\\
&\leq \frac{2}{n-k} \sum_{j=k}^{n-1} \big\Vert X_j(k) \big \Vert  \big\Vert P_{(k)}X_j(k)-\wh P_{(k)}X_j(k) \big\Vert.\label{help6}
\end{align} 
Note that, from the definitions of $X_j(k)$, $P_{(k)}$ and $\wh{P}_{(k)}$, 
\begin{align*}
P_{(k)}X_j(k)
%= (P_{A_{d_k}}X_j,P_{A_{d_{k-1}}}X_{j-1}, \dots,P_{A_{d_1}}X_{j-k})^\top 
&=\bigg(\sum_{i=1}^{d_k}\langle X_j,\nu_i\rangle \nu_i,\ldots,\sum_{i=1}^{d_1}\langle X_{j-k},\nu_i\rangle \nu_i\bigg)^\top, \\
%\label{pkx} 
%\end{align}
%and similarly
%	\begin{align}
\wh P_{(k)}X_j(k)
%= (\wh P_{A_{d_k}}X_j,\wh P_{A_{d_{k-1}}}X_{j-1}, \dots,\wh P_{A_{d_1}}X_{j-k})^\top 
&=\bigg(\sum_{i=1}^{d_k}\langle X_j,\wh \nu_i\rangle \wh\nu_i,\dots,\sum_{i=1}^{d_1}\langle X_{j-k},\wh\nu_i\rangle \wh\nu_i\bigg)^\top.
%\label{whpkx}
\end{align*}
These relations show that 
\begin{align*}
\big\Vert P_{(k)}X_j(k)- \wh P_{(k)}X_j(k)\big\Vert 
&= \bigg\Vert\bigg(\sum_{i=1}^{d_k}\langle X_j,\wh \nu_i\rangle \wh\nu_i - \langle X_j,\nu_i\rangle \nu_i,\dots,\sum_{i=1}^{d_1}\langle X_{j-k},\wh\nu_i\rangle \wh\nu_i-\langle X_{j-k},\nu_i\rangle \nu_i\bigg)^\top\bigg\Vert \\
%	\end{align*}
%	Therefore, using the triangle inequality, we get
%	\begin{align*}	
%\big\Vert P_{(k)}X_j(k)- \wh P_{(k)}X_j(k)\big\Vert	
&= \bigg\Vert\bigg(\sum_{i=1}^{d_k}\langle X_j, \wh\nu_i-\nu_i\rangle \wh\nu_i  ,\dots,\sum_{i=1}^{d_1}\langle X_{j-k},\wh\nu_i-\nu_i\rangle \wh\nu_i\bigg)^\top\bigg\Vert\\
&\qquad+\bigg\Vert\bigg(\sum_{i=1}^{d_k}\langle X_j,\nu_i\rangle (\nu_i-\wh\nu_i),\dots,\sum_{i=1}^{d_1}\langle X_{j-k},\nu_i\rangle (\nu_i-\nu_i)\bigg)^\top\bigg\Vert.
\end{align*}
Observe that, for $x=(x_1,\dots,x_k)\in H^k$,  $\Vert x\Vert = (\sum_{i=1}^{k} \Vert x_i\Vert ^2 )^{1/2}$, Then, applications of the Cauchy--Schwarz inequality and the orthonormality of $(\nu_i\colon i\in\N)$ and $(\wh\nu_i\colon i\in\N)$ lead to %\JK{second summand: show that the squred norm of the sum is the sum of the squared norm!}
\begin{align}
\big\Vert P_{(k)}X_j(k)- \wh P_{(k)}X_j(k)\big\Vert
&\leq \bigg(\sum_{i=0}^{k-1} \bigg\Vert\sum_{i=1}^{d_i} \langle X_{j-i}, \wh\nu_l-\nu_l\rangle \wh\nu_l\bigg\Vert^2 \bigg)^{1/2} 
+ \bigg(\sum_{i=0}^{k-1} \bigg\Vert\sum_{l=1}^{d_i} \langle X_{j-i},\nu_l\rangle (\nu_l-\wh\nu_l)\bigg\Vert^2 \bigg) ^{1/2} \notag\\
&\leq \bigg(\sum_{i=0}^{k-1} \sum_{l=1}^{d_i} \Vert X_{j-i}\Vert^2 \Vert \wh\nu_l-\nu_l\Vert^2 \bigg)^{1/2} 
+ \bigg(\sum_{i=0}^{k-1} \sum_{l=1}^{d_i} \Vert X_{j-i}\Vert^2 \Vert \nu_l-\wh\nu_l \Vert ^2  \bigg) ^{1/2}\notag\\
&\leq 2 \bigg(\sum_{i=0}^{k-1} \sum_{l=1}^{d_k} \Vert X_{j-i}\Vert^2 \Vert \wh\nu_l-\nu_l\Vert^2 \bigg)^{1/2} \notag\\
&\leq 2\Vert X_j(k)\Vert \bigg(  \sum_{l=1}^{d_k} \Vert \wh\nu_l-\nu_l\Vert^2 \bigg)^{1/2}. \notag
%\label{help8}
\end{align}
Plugging this relation back into \eqref{help6}, it follows that
\begin{align}
\Vert \wh{ \Gamma}_{k,d}-\wh {\wh \Gamma}_{k,d}\Vert_\call 
&\leq 4 \bigg(  \sum_{l=1}^{d_k} \Vert \wh\nu_l-\nu_l\Vert^2 \bigg)^{1/2}
\frac{2}{n-k} \sum_{j=k}^{n-1}\Vert X_j(k)\Vert^2. \notag  % \Vert X_j(k) \Vert \notag.
\end{align}
Since $(X_j\colon j\in\Z)$ is $L^4$-$m$ approximable, Theorems~3.1 and 3.2 in H\"ormann and Kokoszka~\cite{weaklydep} imply that, for some finite positive constant $C_1$, $N\E[\Vert \wh\nu_l-\nu_l\Vert^2] \leq C_1/\delta_l$, where $\delta_l$ is the $l$-th spectral gap. Hence, % defined as in Lemma~\ref{lem:specgap}. Hence 
\begin{align*}
\sum_{l=1}^{d_k} \Vert \wh\nu_l-\nu_l\Vert^2  \leq \frac{C_1}{N} \sum_{l=1}^{d_k} \frac{1}{\alpha_l^{2}}.
\end{align*}
Furthermore, note that % an application of the Cauchy--Schwarz inequality yields
\begin{align*}
\frac{2}{n-k}\sum_{j=k}^{n-1} \E\big[\Vert X_j(k)\Vert^2\big] 
%&\leq (N-k)^{-1} \sum_{j=k}^{N-1} 2\, ( E\Vert X_j(k)\Vert^2 ) ^{1/2} (E   \Vert X_j(k) \Vert ^2 )^{1/2}\\
&\leq 2 \sum_{i=0}^{k-1} \E\big[\Vert X_{k-i} \Vert ^2\big] = 2 k \Vert C_X\Vert_\caln.
	\end{align*}
Therefore, collecting the previous results yields the rate
\begin{align}
\Vert \wh{ \Gamma}_{k,d}-\wh {\wh \Gamma}_{k,d}\Vert_\call
= O_p \bigg( \frac{k}{n}\bigg(\sum_{l=1}^{d_k} \frac{1}{\alpha_l^{2}}\bigg)^{1/2}\bigg). 
\label{res4}
\end{align}

Next, investigate $\Vert \wh {\wh \Gamma}_{k,d}^{-1} \Vert_\call$. Similarly as in the corresponding part of the proof of Theorem~\ref{autoregressive}, it follows that
%	\begin{align}
$\Vert \wh {\wh \Gamma}_{k,d}^{-1} \Vert_\call \leq \Vert \wh {\wh \Gamma}_{k,d}^{-1}- {\wh \Gamma}_{k,d}^{-1}\Vert_\call + \Vert \wh { \Gamma}_{k,d}^{-1} \Vert_\call$. %\label{help9}
%\end{align}
By \eqref{res5}, $ \Vert \wh { \Gamma}_{k,d}^{-1} \Vert_\call = O_p(\alpha_{d_k}^{-1})$. Furthermore, the same arguments as in \eqref{help1} and \eqref{help2} imply that
\begin{align}
\Vert \wh {\wh \Gamma}_{k,d}^{-1}- {\wh \Gamma}_{k,d}^{-1}\Vert_\call \leq \frac{\Vert\wh {\wh\Gamma}_{d,k} - \wh\Gamma_{d,k} \Vert_{\call} \Vert  \wh\Gamma_{k,d}^{-1}\Vert_{\call}^{2}}{1-\Vert\wh {\wh\Gamma}_{d,k} - \wh\Gamma_{d,k}\Vert_{\call}\Vert\wh\Gamma_{k,d}^{-1}\Vert_{\call}}. 
\label{help10}
\end{align}
Hence, by \eqref{res5} and \eqref{res4},
\begin{align*}
\Vert\wh {\wh\Gamma}_{d,k} - \wh\Gamma_{d,k} \Vert_{\call} \Vert  \wh\Gamma_{k,d}^{-1}\Vert_{\call}^{2} 
= O_p\bigg(\frac{k}{n\alpha_{d_k}^{2}}\bigg(\sum_{l=1}^{d_k} \frac{1}{\alpha_l^{2}}\bigg)^{1/2}\bigg).
\end{align*} 
Therefore, by Assumption~\ref{ass2} as $n\rightarrow \infty$,
%	\begin{align}
$\Vert \wh {\wh \Gamma}_{k,d}^{-1}- {\wh \Gamma}_{k,d}^{-1}\Vert_\call \overset{p}{\rightarrow} 0$. %\label{res7}
%	\end{align}
Taken the previous calculations together, this gives the rate
\begin{align}
\Vert \wh {\wh \Gamma}_{k,d}^{-1} \Vert_\call = O_p\bigg(\frac{1}{\alpha_{d_k}}\bigg) \label{res8}.
\end{align}
Going back to \eqref{ausgang1} and noticing that $ \Vert\wh{\wh\Gamma}_{1,k,d} -\wh \Gamma_{1,k,d}\Vert_\call \leq \Vert(I_H,0,\dots,0)(\wh{\wh\Gamma}_{k,d} -\wh \Gamma_{k,d})\Vert_\call$, the first summand in this display can be bounded by
\begin{align}
\Vert \big( \wh{\wh\Gamma}_{1,k,d} -\wh \Gamma_{1,k,d} \big) \wh{\wh \Gamma}_{k,d}^{-1}\Vert_\call 
&\leq \Vert  \wh{\wh\Gamma}_{1,k,d} -\wh \Gamma_{1,k,d} \Vert_\call \Vert \wh{\wh \Gamma}_{k,d}^{-1}\Vert_\call\notag \\
&\leq  \Vert  (I_H,0,\dots,0)(\wh{\wh\Gamma}_{k,d} -\wh \Gamma_{k,d}) \Vert_\call \Vert \wh{\wh \Gamma}_{k,d}^{-1}\Vert_\call\notag\\
&=O_p\bigg(\frac{k}{n\alpha_{d_k}}\bigg(\sum_{l=1}^{d_k} \frac{1}{\alpha_l^{2}}\bigg)^{1/2}\bigg), \label{res9}
\end{align}
where the rate in \eqref{res4} was used in the last step. For the second summand in \eqref{ausgang1}, use the plug-in estimator for $\wh\Gamma_{1,k,d}$ to obtain, for all $k<n$,
\begin{align*}
\Vert \wh \Gamma_{1,k,d} \big( \wh{ \Gamma}_{k,d}^{-1}-\wh {\wh \Gamma}_{k,d}^{-1}\big)\Vert_\call 
&\leq \bigg\Vert\frac{1}{n-k} \sum_{j=k}^{n-1} P_{(k)} X_{j}(k) \otimes X_{j+1}\bigg \Vert_\call 
\big\Vert\wh{ \Gamma}_{k,d}^{-1}-\wh {\wh \Gamma}_{k,d}^{-1}\Vert_\call.
\end{align*}
Since 
\begin{align*}
\E\bigg[\bigg\Vert \frac{1}{n-k} \sum_{j=k}^{n-1} P_{(k)} X_{j}(k) \otimes X_{j+1}\bigg\Vert_\call\bigg] 
&\leq \frac{1}{n-k} \sum_{j=k}^{n-1} \E\big[\Vert P_{(k)} X_{j}(k) \otimes X_{j+1} \Vert_\call\big] \\
&\leq \frac{1}{n-k} \sum_{j=k}^{n-1} \big(\E[\Vert P_{(k)} X_{j}(k) \Vert ^2]\big)^{1/2}\big(\E[ \Vert  X_{j+1} \Vert^2]\big)^{1/2} \\ 
&= \bigg( \sum_{l=0}^{k-1} \E[ \Vert X_{j-l} \Vert^2] \bigg)^{1/2} \Vert C_X \Vert_\caln ^{1/2}\\
&= \sqrt{k} \Vert C_X\Vert_{\caln},
\end{align*}
the result in \eqref{help10} implies that
\begin{align}
\big\Vert \wh \Gamma_{1,k,d} \big( \wh{ \Gamma}_{k,d}^{-1}-\wh {\wh \Gamma}_{k,d}^{-1}\big) \big\Vert_\call 
=  O_p\bigg(\frac{k^{3/2}}{n\alpha_{d_k}^{2}}\bigg(\sum_{l=1}^{d_k} \frac{1}{\alpha_l^{2}}\bigg)^{1/2}\bigg). \label{res10}
\end{align}
Applying Assumption~\ref{ass2} to this rate and collecting the results in \eqref{ausgang2}, \eqref{ausgang1}, \eqref{res9} and \eqref{res10}, shows that, for all $x\in H^k$ as $n\rightarrow\infty$,  $\Vert (\wh{\wh B}(k) - \Pi (k)) (x) \Vert \overset{p}{\rightarrow} 0$. This is the claim.

{\em (ii)} Similar to the proof of part {\em (ii)} of Theorem~3.6.
\end{proof}

%%%%%%%%%%%%%%%%%%%%%%%%
%\bibliographystyle{plain}
\bibliography{bibliography}

\begin{thebibliography}{10}

\bibitem{avd}
A.~Aue and A.~Van Delft.
\newblock Testing for stationarity of functional time series in the frequency
  domain.
\newblock {\em Preprint}, 2017.

\bibitem{AHHR}
A.~Aue, S.~H\"ormann, L.~Horv\'ath, and M.~Reimherr.
\newblock Detecting changes in the covariance structure of multivariate time
  series.
\newblock {\em The Annals of Statistics}, 37:4046--4087, 2009.

\bibitem{aue}
A.~Aue, D.~Dubart Norinho, and S.~H\"ormann.
\newblock On the prediction of stationary functional time series.
\newblock {\em Journal of the American Statistical Association}, 110:378--392,
  2015.

\bibitem{bosq}
D.~Bosq.
\newblock {\em Linear Processes in Function Spaces: Theory and Applications}.
\newblock Springer, New York, 2000.

\bibitem{bosq2014}
D.~Bosq.
\newblock Computing the best linear predictor in a {H}ilbert space.
  {A}pplications to general {ARMAH} processes.
\newblock {\em Journal of Multivariate Analysis}, 124:436--450, 2014.

\bibitem{brockwell}
P.J. Brockwell and R.A. Davis.
\newblock {\em Time Series: Theory and Methods (2nd Ed.)}.
\newblock Springer, New York, 1991.

\bibitem{fortet}
R.~Fortet.
\newblock {\em Vecteurs, fonctions et distributions a\`atoires dans les espaces
  de Hilbert}.
\newblock Hermes, Paris, 1995.

\bibitem{gabrys}
R.~Gabrys and P.~Kokoszka.
\newblock Portmanteau test of independence for functional observations.
\newblock {\em Journal of the American Statistical Association},
  102(480):1338--1348, 2007.

\bibitem{hoermann}
S.~H\"ormann, L.~Kidzinski, and M.~Hallin.
\newblock Dynamic functional principal components.
\newblock {\em Journal of the Royal Statistical Society: Series B},
  77:319--348, 2015.

\bibitem{weaklydep}
S.~H\"ormann and P.~Kokoszka.
\newblock Weakly dependent functional data.
\newblock {\em The Annals of Statistics}, 38:1845--1884, 2010.

\bibitem{horvath}
L.~Horv\`ath and P.~Kokoszka.
\newblock {\em Inference for Functional Data with Applications}.
\newblock Springer, New York, 2012.

\bibitem{hsing}
T.~Hsing and R.~Eubank.
\newblock {\em Theoretical Foundations of Functional Data Analysis, with an
  Introduction to Linear Operators}.
\newblock Wiley, West Sussex, UK, 2015.

\bibitem{kk}
J.~Klepsch and C.~Kl\"uppelberg.
\newblock An {I}nnovations {A}lgorithm for the prediction of functional linear
  processes.
\newblock {\em eprint arXiv:1607.05874}, 2016.

\bibitem{KKW}
J.~Klepsch, C.~Kl\"uppelberg, and T.~Wei.
\newblock Prediction of functional {ARMA} processes with an application to
  traffic data.
\newblock {\em Econometrics and Statistics}, 1:128--149, 2016.

\bibitem{kokoreim}
P.~Kokoszka and M.~Reimherr.
\newblock Determining the order of the functional autoregressive model.
\newblock {\em Journal of Time Series Analysis}, 34:116--129, 2013.

\bibitem{lai}
T.L. Lai and C.P. Lee.
\newblock Information and prediction criteria for model selection in stochastic
  regression and {ARMA} models.
\newblock {\em Statistica Sinica}, 7:285--309, 1997.

\bibitem{lewis}
R.~Lewis and G.C. Reinsel.
\newblock Prediction of {M}ultivariate {T}ime {S}eries by {A}utoregressive
  {M}odel {F}itting.
\newblock {\em Journal of Multivariate Analysis}, 16:393--411, 1985.

\bibitem{merlevede}
F.~Merlev\`ede.
\newblock Sur l'inversibilit\'e des processus lin\'eaires \`a valeurs dans un
  espace de {H}ilbert.
\newblock {\em Comptes rendus de l'Acad\'emie des Sciences, S\'erie I},
  321:477--480, 1995.

\bibitem{mitchell2}
H.~Mitchell.
\newblock {\em Topics in Multiple Time Series}.
\newblock PhD thesis, Royal Melbourne Institute of Technology, 1996.

\bibitem{mitchell}
H.~Mitchell and P.J. Brockwell.
\newblock Estimation of the coefficients of a multivariate linear filter using
  the {I}nnovations {A}lgorithm.
\newblock {\em Journal of Time Series Analysis}, 18:157--179, 1997.

\bibitem{nsiri}
S.~Nsiri and R.~Roy.
\newblock On the invertibility of multivariate linear processes.
\newblock {\em Journal of Time Series Analysis}, 14:305--316, 1993.

\bibitem{panaretos}
V.~Panaretos and S.~Tavakoli.
\newblock Fourier analysis of stationary time series in function space.
\newblock {\em The Annals of Statistics}, 41:568--603, 2012.

\bibitem{ramsay1}
J.O. Ramsay and B.W. Silverman.
\newblock {\em Functional Data Analysis (2nd ed.)}.
\newblock Springer Series in Statistics, 2005.

\bibitem{simon}
B.~Simon.
\newblock {\em Operator Theory --- A Comprehensive Course in Analysis, Part 4}.
\newblock AMS, 2015.

\bibitem{spangenberg}
F.~Spangenberg.
\newblock Strictly stationary solutions of {ARMA} equations in {B}anach spaces.
\newblock {\em Journal of Multivariate Analysis}, 121:127--138, 2013.

\bibitem{tsai}
R.S. Tsai.
\newblock {\em Multivariate Time Series Analysis}.
\newblock Wiley, Hoboken, 2014.

\bibitem{turbillon2}
C.~Turbillon, D.~Bosq, J.M. Marion, and B.~Pumo.
\newblock Parameter estimation of moving averages in {H}ilbert spaces.
\newblock {\em Comptes rendus de l'Acad\'emie des Sciences, S\'erie I},
  346:347--350, 2008.

\end{thebibliography}
%%%%%%%%%%%%%%%%%%%%%%%%

\end{document}